\PassOptionsToPackage{hypertexnames=false}{hyperref}
\documentclass[pdflatex,sn-basic]{sn-jnl}

\usepackage{graphicx}%
\usepackage{multirow}%
\usepackage{amsmath,amssymb,amsfonts}%
\usepackage{amsthm}%
\usepackage[title]{appendix}%
\usepackage{xcolor}%
\usepackage{textcomp}%
\usepackage{manyfoot}%
\usepackage{booktabs}%
\usepackage{algorithm}%
\usepackage{algorithmicx}%
\usepackage{algpseudocode}%
\usepackage{listings}%
\usepackage{hhline}
\usepackage{mathtools}
\usepackage{siunitx}
\usepackage{placeins}
\usepackage[capitalise,nameinlink]{cleveref}
\usepackage{hyperref}
\makeatletter
\AddToHook{begindocument/end}{\let\label@in@display\cref@old@label@in@display}
\makeatother

\newcommand{\Prob}{P\,}
\newcommand{\E}{\mathrm{E}}
\newcommand{\OR}{\mathit{R}}
\newcommand{\orcidicon}[1]{\textsuperscript{\href{https://orcid.org/#1}{\includegraphics[height=1em]{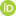}}}}
\theoremstyle{thmstyleone}%
\newtheorem{theorem}{Theorem}

\theoremstyle{thmstyletwo}%

\theoremstyle{thmstylethree}%

\begin{document}

\title{A bivariate cure copula model with zero-inflated gamma frailty for right-censored data}

\author[1]{\fnm{Masaki} \sur{Hino}\orcidicon{0009-0006-7734-6335}}

\author[1,2]{\fnm{Shogo} \sur{Kato}\orcidicon{0000-0003-2648-6769}}

\author*[3]{\fnm{Takeshi} \sur{Emura}\orcidicon{0000-0002-3904-4014}}\email{takeshiemura@gmail.com}

\affil[1]{\orgname{Graduate Institute for Advanced Studies, SOKENDAI}, \orgaddress{\street{10-3 Midori-cho}, \city{Tachikawa}, \postcode{190-8562}, \state{Tokyo}, \country{Japan}}}
 
\affil[2]{\orgname{Institute of Statistical Mathematics}, \orgaddress{\street{10-3 Midori-cho}, \city{Tachikawa}, \postcode{190-8562}, \state{Tokyo}, \country{Japan}}}
 
\affil*[3]{\orgdiv{School of Informatics and Data Science}, \orgname{Hiroshima University}, \orgaddress{\street{1-3-2 Kagamiyama}, \city{Higashi-Hiroshima}, \postcode{739-8511}, \state{Hiroshima}, \country{Japan}}}

\abstract{In biomedical studies, paired survival data arise naturally when two event times are observed within the same subject. Existing statistical models seldom accommodate both cure fractions and complex dependence structures. In this paper, we propose a novel bivariate cure frailty-copula model for paired survival data with cure fractions. By incorporating a zero-inflated gamma frailty, the proposed framework simultaneously accommodates cure fractions and unobserved heterogeneity among uncured subjects. Dependence between cure statuses is modeled naturally via an odds-ratio parameter, while dependence between survival times conditional on frailty is captured through a copula. To provide interpretable measures of overall dependence in the presence of cure, we derive population-level rank correlation coefficients, namely tie-adjusted versions of Kendall's tau and Spearman's rho. For suitable choices of baseline marginal survival functions and copulas, the joint survival function admits a closed-form expression, enabling maximum likelihood estimation and likelihood ratio testing. Simulation studies demonstrate satisfactory finite-sample performance and favorable comparisons with existing bivariate cure models. The diabetic retinopathy data are analyzed for illustration. An \textsf{R} package, \texttt{curecopula}, implementing the proposed methods is publicly available on GitHub.}


\maketitle

\section{Introduction} \label{sec:intro}

In survival analysis, a population may include a subset of subjects who never experience the event of interest; their survival times are effectively infinite and they are treated as cured. Cure models accommodate these long-term survivors by assigning a positive probability mass at infinity under random right censoring \citep{amicoCureModelsSurvival2018,pengCureModelsMethods2021}.

For univariate survival data, cure-frailty models assign positive probability to a frailty value of zero, which yields zero hazard and thus represents cured subjects, whereas positive frailty values capture heterogeneity among susceptible subjects. \citet{priceModellingSurvivalData2001} assumed a continuous gamma distribution for the positive frailty component. Other univariate cure-frailty models have used discrete frailty distributions, such as zero-modified power series distributions \citep{desouzaBayesianCureRate2017,canchoBayesianCureRate2018,molinaSurvivalModelsInduced2021}. Cure models have also been extended to clustered and multivariate time-to-event data via marginal regression approaches \citep{yauLongtermSurvivorMixture2001,yuMixtureCureModels2008,niuSemiparametricMarginalMixture2013,niuMarginalRegressionAnalysis2014,schneiderFreerangingDogsLifetime2022}, where the primary focus is on marginal distributions rather than on the full specification of the joint distribution.

For bivariate survival analysis with full specification of dependence, frailty models, copula models, and their hybrids are central tools \citep{emuraSurvivalAnalysisCorrelated2019,wangLikelihoodbasedInferenceFrailtycopula2020,wangMultivariateFailureTime2021}. The frailty model, introduced by \citet{vaupelImpactHeterogeneityIndividual1979}, captures unobserved heterogeneity and can also induce dependence between event times \citep{hougaardFrailtyModelsSurvival1995,duchateauFrailtyModel2008,haStatisticalModellingSurvival2017,balanTutorialFrailtyModels2020}. Several cure-frailty formulations have been proposed for bivariate survival data \citep{wienkeBivariateFrailtyModel2003,wienkeModellingCureFraction2016,rondeauCureFrailtyModels2013,tawiahBivariateJointFrailty2020}. Copula models provide flexible dependence structures \citep{nelsenIntroductionCopulas2006,joeDependenceModelingCopulas}, and cure-copula approaches have likewise been studied \citep{chatterjeeBivariateCureMixtureApproach2001,louzadaFGMLongTermBivariate2013,deoliveiraperesBivariateLifetimeModels2020,peresBivariateDistributionsBased2021,huangBayesianBivariateCure2022}.

In modeling bivariate survival data with cure fractions, there are three sources of association: dependence between the cure indicators, dependence induced by a shared frailty representing unobserved heterogeneity, and residual dependence not captured by the shared frailty. However, existing models capture only subsets of these components. \citet{chatterjeeBivariateCureMixtureApproach2001} modeled dependence between the cure indicators and dependence between susceptible survival times through a copula, but did not incorporate frailty-induced heterogeneity. \citet{wienkeBivariateFrailtyModel2003} modeled the residual dependence between susceptible survival times by extending the shared frailty to correlated frailties. Although the resulting joint survival function admits a copula representation \citep{wienkeFrailtyModelsSurvival2010}, the copula is induced by the correlated frailty distribution, with dependence determined by the frailty variance and the correlation between the frailties, and therefore cannot be specified independently. \citet{rouzbahaniNewBivariateSurvival2025} combined frailty- and copula-based dependence, but adopted a non-mixture cure formulation in which the cure fractions are determined by the frailty and copula parameters rather than parameterized separately. Under this formulation, the dependence between cure indicators is also induced implicitly by these parameters and is restricted to positive association, rather than being specified independently.

Consequently, none of these existing bivariate approaches separately parameterizes all three sources of association. We therefore propose a bivariate cure copula model with zero-inflated gamma frailty that separates these components. The model uses an odds ratio for dependence between the cure indicators, a positive shared gamma component for unobserved heterogeneity and induced dependence among susceptible subjects, and a copula for residual dependence conditional on frailty. The proposed model provides a broader framework for bivariate cure models by separately parameterizing these components; see \Cref{sec:model} for details.

In addition to separating these dependence components, we derive population-level rank correlation coefficients in the presence of cure. To our knowledge, such coefficients have not previously been explicitly formulated for bivariate cure models. These coefficients summarize, on an interpretable scale, the combined contributions of cure-status dependence, frailty-induced dependence, and residual copula dependence.

For suitable choices of baseline marginal survival functions and copulas, the bivariate survival function admits a closed-form expression. In particular, combining the gamma frailty distribution with Weibull baseline survival functions yields analytically tractable expressions. This enables likelihood-based inference, including maximum likelihood estimation and likelihood ratio testing. To facilitate computation and data analysis, we provide an open-source \textsf{R} package \texttt{curecopula}\footnote{We will make the package available publicly on CRAN after the paper is accepted.} implementing the proposed methods, which is used for all simulation studies and real data analyses in this paper.

This paper is organized as follows. In \Cref{sec:model}, we introduce the proposed cure frailty-copula model with zero-inflated gamma frailty. \Cref{sec:rank} derives population-level rank correlation coefficients in the presence of a cure fraction. \Cref{sec:inference} develops maximum likelihood estimation and likelihood ratio testing. \Cref{sec:software} describes the accompanying \textsf{R} package \texttt{curecopula}. \Cref{sec:sim} presents simulation studies to assess the finite-sample performance of the maximum likelihood estimator. \Cref{sec:realdata} applies the proposed method to the data from the Diabetic Retinopathy Study. Finally, \Cref{sec:conclusion} summarizes the main findings and discusses directions for future research. Online Resource 1 provides technical derivations and additional numerical results.

\section{Proposed model}
\label{sec:model}
In this section, we propose a bivariate cure frailty-copula model. Let $T_j$ $(j=1,2)$ denote two possibly correlated survival times for a subject, where a subject may refer to a single individual (e.g., for recurrent event times or bilateral organ event times) or a pair of individuals (e.g., twins, spouses, or matched pairs). We introduce cure fractions for the distributions of $T_1$ and $T_2$ through a frailty. Let $Y_j\in\{0,1\}$ denote the latent cure indicator for margin $j$, where $Y_j=1$ indicates that margin $j$ is cured and $Y_j=0$ indicates that margin $j$ is susceptible. Thus, the marginal cure probability is $p_j=\Prob(Y_j=1)$.

We adopt a zero-inflated gamma frailty construction, building on the univariate gamma frailty-mixture formulation of \citet{priceModellingSurvivalData2001}; see also \citet{foxMultivariateZeroinflatedModeling2013,wangConfidenceIntervalsZeroinflated2024}. Specifically, we define the multiplicative frailty term in the marginal hazard of $T_j$ as $Z_j = (1-Y_j)W, \, j=1,2$, where a frailty variable $W$ follows a gamma distribution with shape parameter $1/\gamma$ and scale parameter $\gamma$, so that $\E[W]=1$ and $\mathrm{Var}(W)=\gamma$. The density of $W$ is denoted by $f_\gamma$, and its Laplace transform is $\mathcal{L}_W(s)=\E[e^{-sW}]=(1+\gamma s)^{-1/\gamma}$ for $s>0$. Thus, cured margins have $Z_j=0$, whereas uncured margins share the frailty $Z_j=W$. The probability mass of $(Y_1, Y_2)$ is denoted by $p_{rs} = \Prob(Y_1 = r, Y_2 = s)$ for $(r, s) \in \{(0,0), (0,1), (1,0), (1,1)\}$.

Conditional on $Y_1=Y_2=0$ and $W=w$, residual dependence between survival times conditional on the shared frailty is modeled through a copula. Following \citet{emuraJointFrailtycopulaModel2017} (see also \citet{emuraSurvivalAnalysisCorrelated2019,emuraConditionalCopulaModels2021a}), we adopt the joint frailty-copula framework in which the conditional joint survival function given the frailty $W=w$ is
\[
  S(t_1, t_2 \mid w) 
  = C_\theta\left(S_1(t_1 \mid w), S_2(t_2 \mid w)\right),
\]
where $C_\theta$ is a copula with a dependence parameter $\theta$. Under the shared frailty model \citep{duchateauFrailtyModel2008}, we have $S_j(t_j \mid w) = S_{0j}(t_j)^w$, where $S_{0j}(t_j)$ is the baseline survival function for margin $j$.

The joint survival function is obtained by averaging over unobservable frailty terms $(Z_1, Z_2)$. This leads to the following decomposition according to the four possible cure-state configurations:
\begin{align}
  S(t_1, t_2)
  &= p_{11} 
   + p_{01} \int_0^\infty S_{01}(t_1)^w f_\gamma(w)\, \mathrm{d}w
   + p_{10} \int_0^\infty S_{02}(t_2)^w f_\gamma(w)\, \mathrm{d}w \notag \\
  &\quad + p_{00} \int_0^\infty C_\theta\left( S_{01}(t_1)^w, S_{02}(t_2)^w \right) f_\gamma(w)\, \mathrm{d}w,
  \label{eq:joint-surv}
\end{align}
where $S_{0j}(t_j)$ can be specified either parametrically (e.g., Weibull) (\Cref{subsec:closed-form}) or nonparametrically (e.g., splines), possibly including covariates (\Cref{subsec:marginal-cure-cov}). The first term corresponds to both margins being cured, the second and third terms to exactly one margin being cured, and the fourth term to both margins being uncured. A derivation is provided in Section S1 of Online Resource 1.
\vspace{-1em}
\begin{table}[htbp]
\caption{Comparison of dependence structures in the proposed and existing models}
\label{tab:model_structures}
\centering
\setlength{\tabcolsep}{3pt}
\begin{tabular*}{\textwidth}{@{\extracolsep{\fill}}lccc@{}}
\toprule
Model & Cure dependence & \multicolumn{2}{c}{Susceptible survival-time dependence} \\
\cmidrule(lr){3-4}
 & & Frailty & Copula \\
\midrule
Proposed & Odds ratio & Shared frailty & Parametric copula \\
\citet{chatterjeeBivariateCureMixtureApproach2001} & Odds ratio & None & Parametric copula \\
\citet{wienkeBivariateFrailtyModel2003} & Independent & Correlated frailty & Independence copula \\
\citet{rouzbahaniNewBivariateSurvival2025} & Implicit & Shared frailty & Parametric copula \\
\citet{yuMixtureCureModels2008} & Unspecified & Unspecified & Unspecified \\
\citet{niuSemiparametricMarginalMixture2013,niuMarginalRegressionAnalysis2014} & Pearson working correlation & Unspecified$^\dagger$ & Unspecified$^\dagger$ \\
\citet{wangLikelihoodbasedInferenceFrailtycopula2020} & No cure & Shared frailty & Parametric copula \\
\citet{wangMultivariateFailureTime2021} & No cure & Shared frailty & Parametric copula \\
\bottomrule
\end{tabular*}
\begin{minipage}{\textwidth}
  \footnotesize
  Note: ``Implicit'' indicates that cure dependence is determined by other model parameters. $\dagger$ \citet{niuSemiparametricMarginalMixture2013,niuMarginalRegressionAnalysis2014} account for dependence using working correlation matrices in the estimating equations.
\end{minipage}
\vspace{-3\baselineskip}
\end{table}

\Cref{tab:model_structures} compares the proposed model with existing bivariate cure models \citep{chatterjeeBivariateCureMixtureApproach2001,wienkeBivariateFrailtyModel2003,rouzbahaniNewBivariateSurvival2025}, marginal cure models \citep{yuMixtureCureModels2008,niuSemiparametricMarginalMixture2013,niuMarginalRegressionAnalysis2014}, and frailty-copula models without cure \citep{wangLikelihoodbasedInferenceFrailtycopula2020,wangMultivariateFailureTime2021}. The proposed model \eqref{eq:joint-surv} accommodates all three dependence components while encompassing several existing models as special or limiting cases. In the limit $\gamma \to 0$, in which $W$ degenerates at one, the proposed model reduces to the model of \citet{chatterjeeBivariateCureMixtureApproach2001}. By contrast, neither the model of \citet{wienkeBivariateFrailtyModel2003} nor the model of \citet{rouzbahaniNewBivariateSurvival2025} is generally nested within the proposed framework; further details are provided in Section S1.3 of Online Resource 1. The marginal approaches of \citet{yuMixtureCureModels2008} and \citet{niuSemiparametricMarginalMixture2013,niuMarginalRegressionAnalysis2014}  do not specify a full joint distribution. Nevertheless, \citet{niuSemiparametricMarginalMixture2013,niuMarginalRegressionAnalysis2014} account for both types of dependence through working correlation matrices in generalized estimating equations. When the cure fractions approach zero, the proposed model reduces to the frailty-copula model of \citet{wangLikelihoodbasedInferenceFrailtycopula2020} and \citet{wangMultivariateFailureTime2021}.

\subsection{Examples for closed-form expressions}
\label{subsec:closed-form}
As discussed in \Cref{sec:intro}, we adopt Weibull baseline survival functions: $S_{0j}(t_j) = \exp\left(- r_j t_j^{a_j}\right)$ for $t_j > 0$, where $a_j > 0$ is the shape parameter and $r_j > 0$ is the scale parameter \citep{kleinSurvivalAnalysisTechniques2005}. The Weibull distribution is one of the most widely used lifetime distributions in survival analysis because it accommodates increasing, decreasing, and constant hazard functions through a simple parametric form \citep{weibullStatisticalDistributionFunction1951}. It has also been adopted in frailty and copula-based survival models, making it a flexible and convenient choice for modeling heterogeneous and dependent event times \citep{liuPlanningAcceleratedLife2012,rotoloSurrosurvPackageEvaluation2018,schneiderApproachModelClustered2019}.

Under this specification, the second and third terms in \cref{eq:joint-surv} have closed-form expressions by exploiting the Laplace transform $\mathcal{L}_W(s)$ \citep{hougaardAnalysisMultivariateSurvival2000,molenberghsCombinedGammaFrailty2015,wuMetaanalysisIndividualPatient2020}. The availability of a closed-form expression for the fourth term depends on the choice of $C_\theta$. It can be shown that the independence and Gumbel copula models yield closed-form expressions. The Farlie--Gumbel--Morgenstern (FGM) copula also yields a closed-form expression; its model specification and resulting survival function are provided in Section S1.2 of Online Resource 1.

\subsubsection*{Example 1: The independence copula model}
\label{subsec:zig-independence}
Under the independence copula $C(u, v) = uv$, the bivariate survival function is expressed as
\begin{align*}
  S(t_1, t_2)
  &= p_{11} + p_{01} \left(1 + \gamma r_1 t_1^{a_1} \right)^{-1/\gamma} + p_{10} \left(1 + \gamma r_2 t_2^{a_2} \right)^{-1/\gamma} \\
  &\quad + p_{00} \left\{1 + \gamma \left(r_1 t_1^{a_1} + r_2 t_2^{a_2} \right) \right\}^{-1/\gamma}.
\end{align*}
Here, let $S_{j}(t_j) = \left(1 + \gamma r_j t_j^{a_j} \right)^{-1/\gamma}$ for $j = 1, 2$, which are the marginal survival functions obtained by integrating out the shared frailty $W$. These constitute proper survival functions (i.e., $\lim_{t_j \to \infty} S_{j}(t_j) = 0$), and the bivariate survival function can be expressed as follows:
\[
  S(t_1, t_2)
  = p_{11} + p_{01} S_1(t_1) + p_{10} S_2(t_2) + p_{00} \left\{S_1(t_1)^{-\gamma} + S_2(t_2)^{-\gamma} - 1\right\}^{-1/\gamma}.
\]
Thus, the joint survival function of the susceptible subjects is given by a Clayton copula with parameter $\gamma$, where $\gamma$ is the variance of the shared gamma frailty. The independence copula model can also be interpreted as a cure frailty model without an explicitly specified copula, in which the resulting Clayton-type dependence between susceptible survival times is induced entirely by the shared gamma frailty \citep{oakesModelAssociationBivariate1982,hougaardAnalysisMultivariateSurvival2000}.

\subsubsection*{Example 2: The Gumbel copula model}
\label{subsec:zig-gumbel}
The Gumbel copula is defined as
\[
  C_{\theta}(u, v)
  = \exp\left[-\left\{(-\log u)^{\theta + 1} + (-\log v)^{\theta + 1}\right\}^{1/(\theta + 1)}\right],
  \quad \theta \geq 0.
\]
The bivariate survival function is expressed as
\begin{align*}
  S(t_1, t_2)
  &= p_{11} + p_{01} \left(1 + \gamma r_1 t_1^{a_1} \right)^{-1/\gamma} + p_{10} \left(1 + \gamma r_2 t_2^{a_2} \right)^{-1/\gamma} \\
  &\quad + p_{00} \left[1 + \gamma \left\{\left(r_1 t_1^{a_1}\right)^{\theta + 1} + \left(r_2 t_2^{a_2}\right)^{\theta + 1} \right\}^{1/(\theta + 1)}\right]^{-1/\gamma}.
\end{align*}
Setting $\theta = 0$ reduces to the independence copula model in Example 1.

As in Example 1, let $S_j(t_j) = \left(1 + \gamma r_j t_j^{a_j} \right)^{-1/\gamma}$ for $j = 1, 2$. The bivariate survival function can then be expressed as follows:
\begin{align*}
  S(t_1, t_2)
  &= p_{11} + p_{01} S_1(t_1) + p_{10} S_2(t_2) \\
  &\quad + p_{00} \left[1+\left\{\left(S_1(t_1)^{-\gamma} - 1 \right)^{\theta + 1} + \left(S_2(t_2)^{-\gamma} - 1 \right)^{\theta + 1}\right\}^{1/(\theta + 1)}\right]^{-1/\gamma}.
\end{align*}
The term multiplied by $p_{00}$ can be written as a BB1 copula \citep{joeDependenceModelingCopulas}, defined by
\[
  C_{\theta, \gamma} (u, v)
  = \left[1+\left\{\left(u^{-\gamma} - 1 \right)^{\theta + 1} + \left(v^{-\gamma} - 1 \right)^{\theta + 1}\right\}^{1/(\theta + 1)}\right]^{-1/\gamma},
  \quad \theta \geq 0, \, \gamma > 0.
\]
Thus, the joint survival function of the susceptible subjects is given by the BB1 copula with parameters $\theta$ and $\gamma$, where $\theta$ is the dependence parameter of the conditional Gumbel copula and $\gamma$ is the variance of the shared gamma frailty. This can be interpreted as the BB1 copula combining the Clayton-type dependence induced by the shared gamma frailty with the residual dependence modeled by the conditional Gumbel copula.

\subsection{Dependence structure of cure indicators}\label{subsec:or}
Following \citet{chatterjeeBivariateCureMixtureApproach2001}, we quantify the dependence between the cure indicators
$(Y_1,Y_2)$ via an odds ratio, denoted by $\OR$. In contrast, \citet{kimCureRateModel2017} used Pearson's correlation coefficient rather than the odds ratio, whereas \citet{niuSemiparametricMarginalMixture2013,niuMarginalRegressionAnalysis2014} used Pearson-type working correlation matrices in generalized estimating equations. However, for binary variables with covariate-dependent marginal probabilities, a covariate-invariant Pearson correlation coefficient is generally inappropriate because its feasible range depends on the marginals and hence on the covariates \citep{raochagantyEfficiencyGeneralizedEstimating2004}. By contrast, \citet{agrestiIntroductionCategoricalData2007} noted that the odds ratio does not suffer from this limitation and provides a more suitable parametrization of association.

\Cref{tab:joint_cure} displays the $2\times2$ contingency table for the distribution of the cure status $(Y_1, Y_2)$. Here, $p_{rs}$ denotes the cell probability, while $p_1=p_{11}+p_{10}$ and $p_2=p_{11}+p_{01}$ denote the corresponding marginal probabilities. The odds ratio is defined by
\[
  \OR=\frac{p_{11}p_{00}}{p_{10}p_{01}}
  =\frac{\Prob(Y_1=1,Y_2=1)\Prob(Y_1=0,Y_2=0)}
  {\Prob(Y_1=1,Y_2=0)\Prob(Y_1=0,Y_2=1)}.
\]
\vspace{-1em}
\begin{table}[htbp]
  \caption{Joint distribution of $(Y_1,Y_2)$ with row and column margins}
  \label{tab:joint_cure}
  \renewcommand{\arraystretch}{1.5}
  \begin{tabular}{cccc}
    & $Y_2=1$ & $Y_2=0$ & \\ \hhline{~|--|~}
    $Y_1=1$ & \multicolumn{1}{|c}{$p_{11}=\Prob(Y_1=1,Y_2=1)$} & \multicolumn{1}{|c|}{$p_{10}=\Prob(Y_1=1,Y_2=0)$} & $p_1$ \\
    \hhline{~|--|~}
    $Y_1=0$ & \multicolumn{1}{|c}{$p_{01}=\Prob(Y_1=0,Y_2=1)$} & \multicolumn{1}{|c|}{$p_{00}=\Prob(Y_1=0,Y_2=0)$} & $1-p_1$ \\ \hhline{~|--|~}
    & $p_2$ & $1-p_2$ & $1$ \\
  \end{tabular}
\end{table}
\vspace{-2em}

We assume $p_{rs}>0$ for all $r,s\in\{0,1\}$ so that $\OR\in(0,\infty)$ is well-defined. We also consider the limiting case $\OR=\infty$, which can arise when $p_{10}\to0$ or $p_{01}\to0$. Throughout this paper, the notation $\OR=\infty$ refers to the joint limit $p_{10}\to0$ and $p_{01}\to0$, under which $Y_1=Y_2$ almost surely.

If $\OR=1$, then $Y_1$ and $Y_2$ are independent, and hence
$p_{11}=p_1p_2$, $p_{10}=p_1(1-p_2)$, $p_{01}=(1-p_1)p_2$, and $p_{00}=(1-p_1)(1-p_2)$.
If $0<\OR<1$ or $\OR>1$, then $Y_1$ and $Y_2$ are dependent. In this case, we reparametrize
$(p_{11},p_{10},p_{01},p_{00})$ in terms of $(p_1,p_2,\OR)$ as $p_{11}=\{(\OR-1)(p_1+p_2)+1-\sqrt{D}\}/\{2(\OR-1)\}$, where $D=\{(\OR-1)(p_1+p_2)+1\}^2-4\OR(\OR-1)p_1p_2$. The remaining cell probabilities follow directly from the margins as $p_{10}=p_1-p_{11}$, $p_{01}=p_2-p_{11}$, and $p_{00}=1-p_1-p_2+p_{11}$. See Section S2 of Online Resource 1 for the derivation of $p_{11}$ (including the sign choice) and its continuity at $\OR=1$. Thus, $\OR>1$ indicates positive association and $0<\OR<1$ indicates negative association, while $\OR=1$ corresponds to independence. In this sense, $\OR=\infty$ corresponds to perfect concordance.

\subsection{Incorporating covariates}
\label{subsec:marginal-cure-cov}
Following the classical mixture cure regression formulation of \citet{farewellUseMixtureModels1982}, we incorporate the covariate vectors $\boldsymbol{x}_j^{\mathrm{cure}}$ and $\boldsymbol{x}_j^{\mathrm{surv}}$ into the cure probability and the survival-time distribution of susceptible subjects, respectively, for $j=1,2$. The vectors may coincide or differ; when no distinction is needed, either is written simply as $\boldsymbol{x}_j$. Conditional on $\boldsymbol{x}_j$, we model the marginal cure probability by logistic regression,
\begin{equation*}
  \mathrm{logit}\left(\Prob(Y_j=1 \mid \boldsymbol{x}_j)\right)
  = \mathrm{logit}\left(p_j(\boldsymbol{x}_j)\right)
  = \boldsymbol{x}_j^\top\boldsymbol{\beta}_j^{\mathrm{cure}},
\end{equation*}
with regression coefficient vector $\boldsymbol{\beta}_j^{\mathrm{cure}}$, as has commonly been done in the mixture cure framework, including in multivariate and clustered settings \citep{mclachlanRoleFiniteMixture1994,yuMixtureCureModels2008,niuSemiparametricMarginalMixture2013,niuMarginalRegressionAnalysis2014}.

Here, $p_j(\boldsymbol{x}_j)$ is subject-specific, and the population-level cure fraction for margin $j$ is $\E\left[p_j(\boldsymbol{x}_j)\right]$. In the limiting case $\OR=\infty$, the cure indicators coincide almost surely and therefore $p_1(\boldsymbol{x}_1)=p_2(\boldsymbol{x}_2)$, $\boldsymbol{x}_1=\boldsymbol{x}_2\eqqcolon\boldsymbol{x}$ and $\boldsymbol{\beta}_1^{\mathrm{cure}}=\boldsymbol{\beta}_2^{\mathrm{cure}}\eqqcolon\boldsymbol{\beta}^{\mathrm{cure}}$.

Covariates can also be incorporated into the survival function $S_{0j}$ in \cref{eq:joint-surv}. Under a proportional hazards model, $S_{0j}(t_j) = s_{0j}(t_j)^{\exp\left(\boldsymbol{x}_j^\top \boldsymbol{\beta}_j^{\mathrm{surv}}\right)}$, where $\boldsymbol{\beta}_j^{\mathrm{surv}}$ is the regression coefficient vector and $s_{0j}(t_j)$ is the baseline survival function corresponding to $\boldsymbol{x}_j=\boldsymbol{0}$.

\section{Rank correlation coefficients}
\label{sec:rank}
For the proposed model to be useful for practical applications, it should provide a summary measure of dependence. In survival studies, event times are frequently subject to right censoring, which prevents conventional sample rank correlation coefficients from being computed directly from the observed data. Under the proposed model, population-level rank correlation coefficients can instead be estimated by plugging the fitted model parameters into their theoretical expressions, providing an interpretable summary of dependence in the presence of both censoring and cure. To our knowledge, such population-level rank correlation coefficients have not previously been explicitly derived for bivariate cure models. In this section, we derive rank correlation coefficients---Kendall's tau and Spearman's rho---for the proposed model. We further show that, for suitable choices of copulas, Kendall's tau admits closed-form expressions and can attain perfect positive or negative dependence. Explicit formulas are provided for the independence and Gumbel copulas.

Throughout, $\tau_b$ denotes tie-adjusted Kendall's tau \citep{adlerModificationKendallsTau1957,urbanoTreatmentTiesAP2017} and $\rho_b$ denotes tie-adjusted Spearman's rho. These adjusted measures account for the probability mass at infinity in survival times. For the Gumbel copula model, we also confirm that the theoretical values of these coefficients closely match empirical estimates computed from simulated bivariate samples; see Section S4 of Online Resource 1.

Dependence between $T_1$ and $T_2$ is induced jointly by the shared frailty $W$, the copula $C_\theta$ for the uncured subpopulation, and the dependence between the cure indicators $(Y_1,Y_2)$. Accordingly, dependence measures for $(T_1,T_2)$ depend not only on the frailty and copula parameters but also on $(p_1,p_2,\OR)$ through the joint probabilities $p_{rs}$. This applies in particular to the tie-adjusted coefficients $\tau_b$ and $\rho_b$ introduced below. Specifically, the model can attain perfect positive dependence under suitable conditions, while unconditional independence and perfect negative dependence arise as limiting cases when $\gamma \to 0$; see Section S3.3 of Online Resource 1.

For later use, define the copula induced within the uncured subpopulation by
\[
  C_{\theta,\gamma}^\ast(u,v)
  = \int_{0}^{\infty} C_\theta \left(\exp\left\{-w \frac{u^{-\gamma} - 1}{\gamma}\right\}, \exp\left\{-w \frac{v^{-\gamma} - 1}{\gamma}\right\}\right) f_\gamma(w) \, \mathrm{d} w.
\]

\subsection{Kendall's tau}
\label{subsec:kendall}
\begin{theorem}
    \label{thm:kendall}
    Under the proposed model defined in \cref{eq:joint-surv}, the tie-adjusted Kendall's tau is
    \begin{equation*}
        \tau_b = \frac{2(p_{11}p_{00} - p_{01}p_{10})}{\sqrt{(1-p_1^2)(1-p_2^2)}} + \frac{p_{00}^2}{\sqrt{(1-p_1^2)(1-p_2^2)}} \tau_{00} (\theta,\gamma),
    \end{equation*}
    where $\tau_{00} (\theta,\gamma)$ is Kendall's tau among the uncured subpopulation, given by
    \[
      \tau_{00} (\theta,\gamma)
      = 4 \int_0^1 \int_0^1 C_{\theta,\gamma}^\ast (u, v) \, \mathrm{d} C_{\theta,\gamma}^\ast (u, v) - 1.
    \]
\end{theorem}
\vspace{-2em}
\begin{proof}
    See Section S3.1 of Online Resource 1.
\end{proof}
The tie-adjusted Kendall's tau $\tau_b$ is free from the choice of baseline marginal survival functions $S_{0j}(t_j)$. For the independence and Gumbel copulas (Examples 1 and 2 in \Cref{subsec:closed-form}), $\tau_{00}=\gamma/(\gamma+2)$ and $\tau_{00}=1-2/\{(\theta+1)(\gamma+2)\}$, respectively. For the FGM copula model, $\tau_{00}$ has no closed-form expression and must be evaluated numerically. Closed-form expressions for $\tau_b$ under the first two models and attainable ranges under all three are provided in Section S3.3 of Online Resource 1.

\subsection{Spearman's rho} \label{subsec:spearman}

\begin{theorem} \label{thm:spearman}
    Under the proposed model defined in \cref{eq:joint-surv}, the tie-adjusted Spearman's rho is
    \[
        \rho_b = \frac{3(p_{11}p_{00} - p_{01}p_{10})}{\sqrt{(1-p_1^3)(1-p_2^3)}} + \frac{p_{00}(1-p_1)(1-p_2)}{\sqrt{(1-p_1^3)(1-p_2^3)}} \rho_{00} (\theta, \gamma),
    \]
    where $\rho_{00} (\theta, \gamma)$ is Spearman's rho among the uncured subpopulation, given by
    \[
      \rho_{00} (\theta, \gamma)
      = 12 \int_{0}^{1} \int_{0}^{1} C_{\theta,\gamma}^\ast(u,v) \, \mathrm{d} u \, \mathrm{d} v - 3.
    \]
\end{theorem}
\vspace{-2em}
\begin{proof}
    See Section S3.2 of Online Resource 1.
\end{proof}
The tie-adjusted Spearman's rho $\rho_b$ is also free from the choice of baseline marginal survival functions $S_{0j}(t_j)$. For the independence and Gumbel copulas (Examples 1 and 2 in \Cref{subsec:closed-form}), $\rho_{00}$ must be evaluated numerically. Further details and the attainable range of $\rho_b$ are provided in Section S3.3 of Online Resource 1.

\section{Likelihood-based inference}
\label{sec:inference}
This section develops likelihood-based inference for the proposed model with bivariate censored data, including maximum likelihood estimation (MLE) and the likelihood ratio test (LRT) for $H_0: \OR=1$. Let $\boldsymbol{\Psi} = (\theta, \gamma, \boldsymbol{\beta}_1^{\mathrm{cure}}, \boldsymbol{\beta}_2^{\mathrm{cure}}, \OR, \boldsymbol{\phi}_1, \boldsymbol{\phi}_2)$ denote the vector of parameters to be estimated, where $\boldsymbol{\phi}_j$ contains the parameters of the baseline marginal survival function, possibly including $\boldsymbol{\beta}_j^{\mathrm{surv}}$. As shown in \Cref{subsec:marginal-cure-cov}, the marginal cure fraction $p_j(\boldsymbol{x}_{ij})$ depends on covariates through parameters $\boldsymbol{\beta}_j^{\mathrm{cure}}$. Furthermore, as shown in \Cref{subsec:or}, the joint probabilities $p_{rs,i}$ $(r,s \in \{0,1\})$ can be expressed as functions of $p_1(\boldsymbol{x}_{i1})$, $p_2(\boldsymbol{x}_{i2})$, and $\OR$, and thus become subject-specific through the covariates. For simplicity, when covariates are not incorporated, the parameter vector reduces to $\boldsymbol{\Psi}^\ast = (\theta, \gamma, p_1, p_2, \OR, \boldsymbol{\phi}_1, \boldsymbol{\phi}_2)$ and the joint probabilities $p_{rs}$ are common across subjects.

Suppose we have $n$ independent subjects. For subject $i$ and margin $j$, let $t_{ij} = \min(T_{ij}, C_{ij})$ be the observed time and $\delta_{ij} = I(T_{ij} \leq C_{ij})$ the event indicator, where $C_{ij}$ is the censoring time and $I(\cdot)$ denotes the indicator function. The observed data also include the corresponding covariates, generically denoted by $\boldsymbol{x}_{ij}$. We assume that $(T_{i1}, T_{i2})$ and $(C_{i1}, C_{i2})$ are independent conditional on the covariates $\boldsymbol{x}_{ij}$ and frailty $W$. The log-likelihood function is given by
\begin{align*}
    \ell(\boldsymbol{\Psi})
    &= \sum_{i=1}^{n} \delta_{i1}(1 - \delta_{i2}) \log\left\{- \frac{\partial S(t_{i1}, t_{i2})}{\partial t_{i1}}\right\}
    + (1 - \delta_{i1})\delta_{i2} \log\left\{- \frac{\partial S(t_{i1}, t_{i2})}{\partial t_{i2}}\right\} \\
    &\quad + \delta_{i1} \delta_{i2} \log\left\{\frac{\partial^2 S(t_{i1}, t_{i2})}{\partial t_{i1} \partial t_{i2}}\right\}
    + (1 - \delta_{i1})(1 - \delta_{i2}) \log\left\{S(t_{i1}, t_{i2})\right\}.
\end{align*}
Explicit log-likelihood expressions for the independence, Gumbel, and FGM copula models with Weibull margins are provided in Section S5 of Online Resource 1.

\subsection{Maximum likelihood estimation} \label{subsec:mle}
The maximum likelihood estimator is obtained by direct numerical maximization of $\ell(\boldsymbol{\Psi})$. Because the likelihood is parameterized differently across the four regimes $0<\OR<1$, $\OR=1$, $1<\OR<\infty$, and $\OR=\infty$, maximization is performed separately within each regime, and the solution attaining the largest log-likelihood is selected. Alternatively, the cure indicators $(Y_1,Y_2)$ may be treated as missing data and the resulting complete-data likelihood maximized using the EM algorithm; a computational comparison with direct optimization is provided in Section S7 of Online Resource 1.

Standard errors (SEs) are computed from the inverse observed information matrix. Wald-type confidence intervals (CIs) are constructed using the asymptotic normal approximation. Details of the parameter dimensions, asymptotic approximation, parameter transformations, and CI construction are provided in Section S5.4 of Online Resource 1.

\subsection{\texorpdfstring{Likelihood ratio test for $H_0:\OR=1$}{Likelihood ratio test for H0: R=1}}
\label{subsec:lrt}
We propose a test for the independence hypothesis $H_0:\OR=1$ against $H_1:\OR\neq 1$ using the likelihood ratio statistic $\lambda
=2\left\{\sup_{\boldsymbol{\Psi}\in\Theta}\ell(\boldsymbol{\Psi})
-\sup_{\boldsymbol{\Psi}\in\Theta_0}\ell(\boldsymbol{\Psi})\right\}$, where $\Theta$ denotes the parameter space of the unrestricted model (with $\OR>0$), and $\Theta_0$ denotes the parameter space under $H_0:\OR=1$. In computation, the maximization over $\Theta$ is carried out by maximizing the corresponding log-likelihood separately over the four regimes $\OR\in(0,1)$, $\OR=1$, $\OR\in(1,\infty)$, and $\OR=\infty$ (the likelihood is continuous at $\OR=1$ but is expressed differently across regimes), and then taking the largest maximized value. Under standard regularity conditions, $\lambda \overset{\mathrm{d}}{\to} \chi^2_1$ under $H_0$, and we reject $H_0$ at significance level $\alpha$ if $\lambda>\chi^2_{1,\,\alpha}$, where $\chi^2_{1,\,\alpha}$ denotes the upper $\alpha$ quantile of the chi-squared distribution with one degree of freedom. The finite-sample type I error and power of the proposed likelihood ratio test are investigated via simulation in \Cref{subsec:sim-lrt}.

\section{Software implementation} \label{sec:software}

The proposed methodology is implemented in an \textsf{R} package \texttt{curecopula}, publicly available on GitHub. The package enables random number generation from the proposed cure frailty-copula models based on the independence, Gumbel, and FGM copulas, with or without covariates. Censored data can also be generated. The package also performs maximum likelihood estimation for each model and returns parameter estimates together with inferential quantities such as SEs based on the observed information matrix. Maximum likelihood estimation is implemented by direct numerical maximization of the log-likelihood using the BFGS algorithm via \texttt{optim} in \textsf{R}. To improve numerical stability, the optimization is performed from multiple starting values obtained by systematic perturbations of the initial values, and the converged solution with the largest log-likelihood is retained. The Hessian is evaluated at the maximizer, and SEs are obtained from the inverse observed information matrix. The package also reports the maximized log-likelihood value, enabling likelihood-based inference. For example, likelihood ratio tests for assessing the independence of cure indicators (i.e., $\OR=1$) can be conducted by comparing the fitted models under the null and alternative hypotheses. In addition, the population-level rank correlation coefficients derived in this study can be computed by plug-in estimation using fitted parameter values (see \Cref{sec:rank} for details).

\section{Simulation study} \label{sec:sim}

In this section, we evaluate the finite-sample performance of the proposed method and compare it with the existing methods of \citet{chatterjeeBivariateCureMixtureApproach2001}, \citet{wienkeBivariateFrailtyModel2003}, and \citet{niuSemiparametricMarginalMixture2013,niuMarginalRegressionAnalysis2014}, hereafter referred to as the CS, WLY, and NP models, respectively. \Cref{subsec:sim-design} describes the simulation designs, and \Cref{subsec:sim-results} presents the corresponding results. All simulations were performed using R version 4.5.1 on a Mac Studio equipped with an Apple M4 Max chip and 64 GB of RAM.

\subsection{Simulation design} \label{subsec:sim-design}

\noindent\textbf{Finite-sample performance of the proposed method}\par
\noindent
We evaluate the MLE for the proposed model and the LRT for $H_0:\OR=1$ under the Gumbel copula model with Weibull margins. For $i=1,\ldots,n$ and $j=1,2$, we generate $x_{ij}^{\mathrm{cure}}\sim\mathrm{Bernoulli}(0.5)$ and $x_{ij}^{\mathrm{surv}}\sim\mathrm{Uniform}(0,1)$, independently across subjects, margins, and components. Survival times $(T_{i1}, T_{i2})$ are generated given $(x_{i1}^{\mathrm{cure}}, x_{i1}^{\mathrm{surv}}, x_{i2}^{\mathrm{cure}}, x_{i2}^{\mathrm{surv}})$, following the proposed model with the Gumbel copula. The detailed data generation algorithm is provided in Section S6.1 of Online Resource 1.

We consider sample sizes $n \in \{200, 400\}$ and use the following parameter values: $(\gamma, a_1, r_1, a_2, r_2) = (0.5, 1, 1.5, 1, 2)$ and $(\beta_{01}^{\mathrm{cure}}, \beta_{02}^{\mathrm{cure}}, \beta_{1}^{\mathrm{cure}}, \beta_{2}^{\mathrm{cure}}, \beta_{1}^{\mathrm{surv}}, \beta_{2}^{\mathrm{surv}}) = (1, -1, -1, 1, -1, 1)$, where $\beta_{0j}^{\mathrm{cure}}$ denotes the cure-component intercept for margin $j$. For evaluating the MLE, we set $(\theta, \OR) = (2, 2)$. For evaluating the LRT, we consider $\theta=2$ and $\theta=0.5$, representing stronger and weaker conditional Gumbel-copula dependence, respectively, and $\OR \in \{0.25, 0.333, 0.5, 0.667, 0.75, 0.85, 1, 1.2, 1.5, 2, 3, 4, 8, 16\}$, representing a range of negative and positive associations between the cure indicators, with $\OR=1$ corresponding to independence.

A common independent censoring time $C_{i1}=C_{i2}=C_i$ is drawn from $\mathrm{Uniform}(0,6)$ for the two margins. The average censoring rates are approximately $72\%$ and $44\%$ for margins 1 and 2, respectively. These rates are reasonably close to those observed in the real data analyzed in \Cref{sec:realdata} ($72.6\%$ and $48.7\%$, respectively), indicating that the simulation reflects a realistic censoring setting.

We use $1{,}000$ Monte Carlo replications for each simulation scenario. Estimation performance is assessed in terms of bias, mean squared error (MSE), empirical standard deviation (SD), the mean of the estimated SEs (SE (mean)), and coverage probability (CP). The performance of the LRT is assessed using empirical rejection probabilities at the $5\%$ significance level for each combination of $n$, $\theta$, and $\OR$.

\vspace{0.5em}
\noindent\textbf{Comparison with existing models}\par
\noindent
To compare estimation performance, we generate data with covariates in both the cure and susceptible survival-time components under the proposed, CS, WLY, and NP models. We fit all four models to each dataset, using the \textsf{R} package \texttt{geecure} \citep[version 1.0-6]{niuGeecureMarginalProportional2018} for the NP model. For comparability, we extend the CS and WLY models summarized in \Cref{tab:model_structures} to incorporate covariate effects on both cure probabilities and survival times, as in the proposed model. We also extend the WLY model to accommodate cure dependence through an odds ratio.

For each data-generating model, we consider sample sizes $n\in\{200,400\}$. We set $\OR=2$ and use Weibull baseline marginal survival functions with $(a_1,r_1,a_2,r_2)=(1,1.5,1,2)$. We use $(\theta,\gamma)=(0.6,0.5)$ for the proposed model, $\theta=1$ for the CS model, and frailty correlation and variance equal to $0.8$ and $0.5$, respectively, for the WLY model. For the NP model, we adopt the data-generating mechanism used in the simulation studies of \citet{niuSemiparametricMarginalMixture2013,niuMarginalRegressionAnalysis2014}, specifying a Clayton copula with parameter $2$. These parameter settings yield Kendall's $\tau_{00}=0.5$, conditional on covariates, for the proposed, CS, and NP models, and the same marginal censoring rate for the proposed and WLY models. We generate $1,000$ replications with a common censoring time drawn from $\mathrm{Uniform}(0,6)$. The estimands of interest are the regression coefficients and population-level Kendall's $\tau_b$. For each estimand, we report relative bias (RB), root mean squared error (RMSE), and empirical SD. If convergence issues arise during estimation, we remove these replications and report the proportion of such issues (the issues occured only in the NP models).

\subsection{Results}
\label{subsec:sim-results}

The results are organized as follows. \Cref{subsec:sim-finite-sample} evaluates the finite-sample performance of the MLE, \Cref{subsec:sim-lrt} examines the empirical size and power of the LRT, and \Cref{subsec:sim-comparison} compares the proposed model with the existing models.

\subsubsection{Maximum likelihood estimation}
\label{subsec:sim-finite-sample}

Based on the results reported in \Cref{tab:cov_param}, the proposed estimator demonstrates satisfactory finite-sample performance. Increasing the sample size from $n=200$ to $n=400$ reduces the absolute biases and MSEs for all parameters, consistent with the asymptotic theory of maximum likelihood estimation. The empirical SDs and mean estimated SEs are generally in close agreement, particularly when $n=400$. Coverage probabilities range from $0.912$ to $0.964$ for $n=200$ and from $0.930$ to $0.957$ for $n=400$, remaining reasonably close to the nominal $95\%$ level. Overall, these results support the validity of the proposed method and show that the covariate effects in both model components can be estimated reliably.

\begin{table}[htbp]
  \caption{Estimation performance for the model with covariates in both the cure and susceptible survival-time components ($\theta=2$, $\OR = 2$; $n=200, 400$) based on $1,000$ replications}
  \label{tab:cov_param}
  \centering
  \scriptsize
  \renewcommand{\arraystretch}{1.3}
  \begin{tabular*}{\textwidth}{@{\extracolsep{\fill}}c c S S S c c c@{}}
    \toprule
    {} & Parameter & {True} & {Bias} & {MSE} & {SD} & {SE (mean)} & {CP} \\
    \midrule
    \multirow[t]{13}{*}{$n=200$}
    & $\theta$      & 2.0  & 0.147  & 0.430 & 0.639 & 0.599 & 0.931 \\
    & $\gamma$      & 0.5  & 0.111  & 0.314 & 0.549 & 0.455 & 0.927 \\
    & $\OR$         & 2.0  & 0.219  & 1.256 & 1.100 & 0.975 & 0.949 \\
    & $a_1$         & 1.0  & 0.039  & 0.029 & 0.167 & 0.153 & 0.928 \\
    & $r_1$         & 1.5  & 0.211  & 0.653 & 0.780 & 0.632 & 0.934 \\
    & $a_2$         & 1.0  & 0.036  & 0.028 & 0.162 & 0.148 & 0.912 \\
    & $r_2$         & 2.0  & 0.330  & 1.512 & 1.185 & 0.934 & 0.931 \\
    & $\beta_{01}^{\mathrm{cure}}$ & 1.0  & 0.017  & 0.083 & 0.288 & 0.285 & 0.961 \\
    & $\beta_{02}^{\mathrm{cure}}$ & -1.0 & -0.076 & 0.419 & 0.644 & 0.614 & 0.964 \\
    & $\beta_{1}^{\mathrm{cure}}$  & -1.0 & -0.037 & 0.147 & 0.382 & 0.372 & 0.947 \\
    & $\beta_{2}^{\mathrm{cure}}$  & 1.0  & 0.061  & 0.417 & 0.643 & 0.663 & 0.953 \\
    & $\beta_{1}^{\mathrm{surv}}$  & -1.0 & -0.029 & 0.162 & 0.402 & 0.367 & 0.939 \\
    & $\beta_{2}^{\mathrm{surv}}$  & 1.0  & 0.044  & 0.148 & 0.382 & 0.339 & 0.928 \\
    \midrule
    \multirow[t]{13}{*}{$n=400$}
    & $\theta$      & 2.0  & 0.080  & 0.191 & 0.430 & 0.415 & 0.932 \\
    & $\gamma$      & 0.5  & 0.046  & 0.125 & 0.350 & 0.317 & 0.949 \\
    & $\OR$         & 2.0  & 0.113  & 0.467 & 0.674 & 0.655 & 0.952 \\
    & $a_1$         & 1.0  & 0.016  & 0.012 & 0.108 & 0.106 & 0.941 \\
    & $r_1$         & 1.5  & 0.074  & 0.171 & 0.407 & 0.371 & 0.940 \\
    & $a_2$         & 1.0  & 0.012  & 0.011 & 0.105 & 0.102 & 0.933 \\
    & $r_2$         & 2.0  & 0.115  & 0.358 & 0.587 & 0.535 & 0.930 \\
    & $\beta_{01}^{\mathrm{cure}}$ & 1.0  & 0.007  & 0.042 & 0.205 & 0.198 & 0.933 \\
    & $\beta_{02}^{\mathrm{cure}}$ & -1.0 & -0.032 & 0.045 & 0.211 & 0.206 & 0.957 \\
    & $\beta_{1}^{\mathrm{cure}}$  & -1.0 & -0.024 & 0.076 & 0.275 & 0.259 & 0.935 \\
    & $\beta_{2}^{\mathrm{cure}}$  & 1.0  & 0.030  & 0.061 & 0.246 & 0.242 & 0.948 \\
    & $\beta_{1}^{\mathrm{surv}}$  & -1.0 & -0.012 & 0.065 & 0.254 & 0.249 & 0.951 \\
    & $\beta_{2}^{\mathrm{surv}}$  & 1.0  & 0.011  & 0.058 & 0.242 & 0.232 & 0.947 \\
    \bottomrule
  \end{tabular*}
\end{table}

\subsubsection{\texorpdfstring{Likelihood ratio test for $H_0:\OR=1$}{Likelihood ratio test for H0: R=1}}
\label{subsec:sim-lrt}

We evaluate the finite-sample performance of the likelihood ratio test for $H_0:\OR=1$ in terms of type~I error rate and power.

\begin{figure}[htbp]
    \begin{center}
      \includegraphics[width=0.9\textwidth]{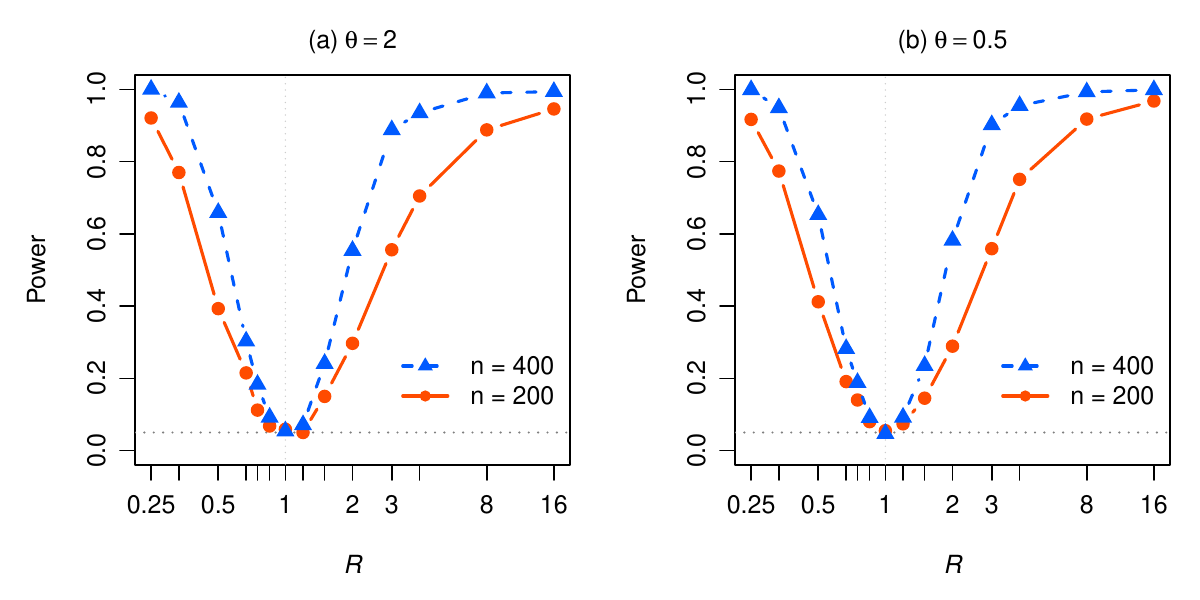}
    \end{center}
    \setlength{\abovecaptionskip}{-1em}
    \caption{Empirical power of the likelihood ratio test for $H_0:\OR=1$ at the $5\%$ significance level under two levels of conditional Gumbel-copula dependence, based on $1{,}000$ Monte Carlo replications. Panels (a) and (b) correspond to stronger dependence ($\theta = 2$) and weaker dependence ($\theta = 0.5$), respectively. Power is shown as a function of $\OR$ (log scale) for $n=200$ and $n=400$. The point at $\OR=1$ represents the empirical type I error under $H_0$, and the horizontal dotted line indicates the nominal level $\alpha=0.05$}
    \label{fig:sim_lrt_power}
\end{figure}

\Cref{fig:sim_lrt_power} visualizes the empirical power of the LRT at the $5\%$ level. At $\OR=1$, the empirical rejection rates are close to the nominal $5\%$ level, confirming that the type~I error rate is well controlled. As expected, the power of the LRT increases with $n$ and as $\OR$ departs from $1$. For example, the power at $\OR=1/4$ was higher than that at $\OR=4$. The power was nearly one at $\OR=1/4$, whereas at $\OR=4$ it was approximately 0.7 for $n=200$ and 0.9 for $n=400$. The overall power pattern is similar between the stronger ($\theta=2$) and weaker ($\theta=0.5$) conditional dependences, suggesting that, over the range considered, the power is not sensitive to the value of $\theta$.

\subsubsection{Comparison with existing models}
\label{subsec:sim-comparison}

First, we compare computation times and convergence rates. The NP model required considerably longer computation times and exhibited some convergence issues (\Cref{tab:computational-performance}). Convergence rates for the NP model were lower under frailty-based data-generating models than under models without frailty, at approximately $75\%$ versus $95\%$ for $n=200$ and $80-85\%$ versus $99\%$ for $n=400$. The convergence issues are caused by the excess from the upper limit of the \textsf{R} package \texttt{geecure}.

\begin{table}[htbp]
\centering
\renewcommand{\tablecaptionfont}{\small}
\caption{Median computation times and convergence rates}
\label{tab:computational-performance}
\small
\setlength{\tabcolsep}{2pt}
\renewcommand{\arraystretch}{0.95}
\begin{tabular*}{\textwidth}{@{\extracolsep{\fill}}lc*{8}{r}@{}}
\toprule
 & & \multicolumn{2}{c}{Prop.}
   & \multicolumn{2}{c}{CS}
   & \multicolumn{2}{c}{WLY}
   & \multicolumn{2}{c}{NP} \\
\cmidrule(lr){3-4}
\cmidrule(lr){5-6}
\cmidrule(lr){7-8}
\cmidrule(lr){9-10}
DGM & $n$
 & \shortstack{Time\\(seconds)} & \shortstack{Conv.\\rate}
 & \shortstack{Time\\(seconds)} & \shortstack{Conv.\\rate}
 & \shortstack{Time\\(seconds)} & \shortstack{Conv.\\rate}
 & \shortstack{Time\\(seconds)} & \shortstack{Conv.\\rate} \\
\midrule
\multirow{2}{*}{Prop.}
 & 200 & 6.6 & 1.000 & 4.7 & 1.000 & 6.7 & 1.000 & 258.6 & 0.752 \\
 & 400 & 9.3 & 1.000 & 6.1 & 1.000 & 8.6 & 1.000 & 606.0 & 0.836 \\
\addlinespace[3pt]
\multirow{2}{*}{CS}
 & 200 & 7.2 & 1.000 & 4.5 & 1.000 & 6.5 & 1.000 & 184.1 & 0.942 \\
 & 400 & 11.1 & 1.000 & 6.0 & 1.000 & 8.6 & 1.000 & 423.4 & 0.988 \\
\addlinespace[3pt]
\multirow{2}{*}{WLY}
 & 200 & 8.0 & 1.000 & 5.2 & 1.000 & 6.8 & 1.000 & 216.8 & 0.775 \\
 & 400 & 10.5 & 1.000 & 6.3 & 1.000 & 8.2 & 1.000 & 505.8 & 0.800 \\
\addlinespace[3pt]
\multirow{2}{*}{NP}
 & 200 & 7.9 & 1.000 & 5.0 & 1.000 & 6.9 & 1.000 & 169.9 & 0.952 \\
 & 400 & 10.4 & 1.000 & 6.4 & 1.000 & 9.2 & 1.000 & 399.6 & 0.987 \\
\bottomrule
\end{tabular*}

\smallskip
\begin{minipage}{\textwidth}
\small
\textit{Note:} DGM, data-generating model; Prop., proposed model; Conv. rate, convergence rate. Time denotes the median wall-clock time in seconds per model fit across all 1,000 Monte Carlo replications, including nonconvergent fits.
\end{minipage}
\end{table}

\Cref{tab:sim_model_comparison} compares the performance of the proposed, CS, WLY, and NP estimators when the data are generated from each of the four competing models. Because the results for the NP estimator are limited to converged replications, they should be regarded as reference values.

\begin{sidewaystable}[htbp]
\caption{Comparison of the proposed, CS, WLY, and NP estimators under four data-generating models based on $1{,}000$ replications}
\label{tab:sim_model_comparison}
\centering
\scriptsize
\setlength{\tabcolsep}{3pt}
\renewcommand{\arraystretch}{1.05}
\newcommand{\missingresult}{\multicolumn{1}{c}{--}}
\newcommand{\resultplaceholders}{\missingresult & \missingresult & \missingresult & \missingresult & \missingresult & \missingresult & \missingresult & \missingresult & \missingresult & \missingresult & \missingresult & \missingresult}
\begin{tabular}{@{}lccr *{3}{rrrr}@{}}
\toprule
& & & & \multicolumn{4}{c}{RB} & \multicolumn{4}{c}{RMSE} & \multicolumn{4}{c}{SD} \\
\cmidrule(lr){5-8}\cmidrule(lr){9-12}\cmidrule(lr){13-16}
DGM & Estimand & True & \multicolumn{1}{c}{$n$} & Prop. & CS & WLY & NP$^{\dagger}$ & Prop. & CS & WLY & NP$^{\dagger}$ & Prop. & CS & WLY & NP$^{\dagger}$ \\
\midrule
\multirow{10}{*}{Prop.}
& $\beta_1^{\mathrm{cure}}$ & -1.000 & 200 & \textbf{0.081} & 0.083 & 0.372 & -0.131 & 0.762 & 0.830 & 1.783 & \textbf{0.435} & 0.758 & 0.827 & 1.745 & \textbf{0.415} \\
& & & 400 & \textbf{0.032} & \textbf{0.032} & 0.127 & -0.067 & \textbf{0.321} & 0.324 & 0.457 & 0.341 & \textbf{0.320} & 0.322 & 0.439 & 0.334 \\
& $\beta_2^{\mathrm{cure}}$ & 1.000 & 200 & 0.116 & \textbf{-0.031} & 0.605 & -0.065 & 1.015 & 0.551 & 2.392 & \textbf{0.347} & 1.009 & 0.551 & 2.315 & \textbf{0.341} \\
& & & 400 & \textbf{0.036} & -0.062 & 0.244 & -0.076 & 0.489 & \textbf{0.252} & 0.582 & 0.254 & 0.488 & 0.244 & 0.529 & \textbf{0.242} \\
& $\beta_1^{\mathrm{surv}}$ & -1.000 & 200 & \textbf{0.066} & -0.103 & 0.355 & -0.283 & 0.431 & \textbf{0.362} & 0.677 & 0.497 & 0.426 & \textbf{0.347} & 0.577 & 0.409 \\
& & & 400 & \textbf{0.009} & -0.146 & 0.260 & -0.256 & 0.284 & \textbf{0.276} & 0.455 & 0.407 & 0.284 & \textbf{0.235} & 0.373 & 0.317 \\
& $\beta_2^{\mathrm{surv}}$ & 1.000 & 200 & \textbf{0.023} & -0.228 & 0.313 & -0.234 & \textbf{0.306} & 0.315 & 0.501 & 0.345 & 0.305 & \textbf{0.218} & 0.391 & 0.254 \\
& & & 400 & \textbf{0.018} & -0.230 & 0.305 & -0.238 & \textbf{0.216} & 0.276 & 0.414 & 0.296 & 0.215 & \textbf{0.152} & 0.279 & 0.176 \\
& $\tau_b$ & 0.155 & 200 & \textbf{0.027} & 0.088 & -0.125 & \missingresult & \textbf{0.073} & 0.074 & 0.080 & \missingresult & 0.073 & \textbf{0.072} & 0.077 & \missingresult \\
& & & 400 & \textbf{-0.006} & 0.049 & -0.193 & \missingresult & 0.052 & \textbf{0.051} & 0.059 & \missingresult & 0.052 & \textbf{0.050} & 0.051 & \missingresult \\
\midrule
\multirow{10}{*}{CS}
& $\beta_1^{\mathrm{cure}}$ & -1.000 & 200 & 0.020 & 0.019 & 0.186 & \textbf{0.002} & \textbf{0.376} & \textbf{0.376} & 1.016 & 0.396 & \textbf{0.375} & \textbf{0.375} & 0.999 & 0.396 \\
& & & 400 & 0.010 & \textbf{0.009} & 0.085 & 0.012 & \textbf{0.271} & \textbf{0.271} & 0.333 & 0.304 & \textbf{0.271} & \textbf{0.271} & 0.322 & 0.304 \\
& $\beta_2^{\mathrm{cure}}$ & 1.000 & 200 & 0.012 & \textbf{0.008} & 0.099 & \textbf{0.008} & 0.320 & 0.319 & 0.384 & \textbf{0.317} & 0.319 & 0.319 & 0.371 & \textbf{0.317} \\
& & & 400 & 0.013 & \textbf{0.011} & 0.083 & \textbf{0.011} & \textbf{0.236} & \textbf{0.236} & 0.273 & 0.237 & \textbf{0.236} & \textbf{0.236} & 0.260 & 0.237 \\
& $\beta_1^{\mathrm{surv}}$ & -1.000 & 200 & 0.057 & 0.027 & 0.349 & \textbf{-0.004} & 0.288 & \textbf{0.273} & 0.570 & 0.379 & 0.282 & \textbf{0.272} & 0.450 & 0.379 \\
& & & 400 & 0.025 & 0.004 & 0.309 & \textbf{0.002} & 0.203 & \textbf{0.194} & 0.436 & 0.274 & 0.201 & \textbf{0.194} & 0.307 & 0.274 \\
& $\beta_2^{\mathrm{surv}}$ & 1.000 & 200 & 0.056 & \textbf{0.016} & 0.360 & 0.023 & 0.220 & \textbf{0.191} & 0.506 & 0.218 & 0.212 & \textbf{0.190} & 0.355 & 0.217 \\
& & & 400 & 0.043 & \textbf{0.017} & 0.342 & 0.021 & 0.147 & \textbf{0.131} & 0.410 & 0.154 & 0.140 & \textbf{0.130} & 0.225 & 0.153 \\
& $\tau_b$ & 0.155 & 200 & \textbf{0.017} & 0.018 & -0.154 & \missingresult & \textbf{0.065} & \textbf{0.065} & 0.072 & \missingresult & \textbf{0.065} & \textbf{0.065} & 0.068 & \missingresult \\
& & & 400 & \textbf{0.019} & \textbf{0.019} & -0.149 & \missingresult & \textbf{0.046} & \textbf{0.046} & 0.051 & \missingresult & \textbf{0.046} & \textbf{0.046} & \textbf{0.046} & \missingresult \\
\midrule
\multirow{10}{*}{WLY}
& $\beta_1^{\mathrm{cure}}$ & -1.000 & 200 & 0.461 & 0.372 & 0.637 & \textbf{-0.106} & 2.320 & 2.087 & 2.752 & \textbf{0.443} & 2.275 & 2.054 & 2.679 & \textbf{0.431} \\
& & & 400 & 0.168 & 0.112 & 0.230 & \textbf{-0.061} & 1.225 & 1.201 & 1.312 & \textbf{0.317} & 1.214 & 1.197 & 1.292 & \textbf{0.311} \\
& $\beta_2^{\mathrm{cure}}$ & 1.000 & 200 & 0.093 & \textbf{-0.058} & 0.327 & -0.077 & 1.166 & 0.369 & 1.813 & \textbf{0.366} & 1.163 & 0.365 & 1.785 & \textbf{0.358} \\
& & & 400 & \textbf{0.008} & -0.066 & 0.055 & -0.076 & 0.271 & 0.251 & 0.506 & \textbf{0.247} & 0.271 & 0.242 & 0.503 & \textbf{0.236} \\
& $\beta_1^{\mathrm{surv}}$ & -1.000 & 200 & \textbf{0.024} & -0.118 & 0.103 & -0.292 & 0.596 & 0.554 & 0.668 & \textbf{0.498} & 0.596 & 0.542 & 0.660 & \textbf{0.404} \\
& & & 400 & \textbf{0.023} & -0.128 & 0.083 & -0.223 & 0.373 & 0.373 & 0.404 & \textbf{0.359} & 0.373 & 0.351 & 0.396 & \textbf{0.282} \\
& $\beta_2^{\mathrm{surv}}$ & 1.000 & 200 & \textbf{0.008} & -0.221 & 0.119 & -0.206 & 0.338 & 0.336 & 0.437 & \textbf{0.327} & 0.338 & \textbf{0.253} & 0.421 & 0.255 \\
& & & 400 & \textbf{-0.024} & -0.252 & 0.045 & -0.236 & \textbf{0.220} & 0.307 & 0.253 & 0.298 & 0.219 & \textbf{0.175} & 0.249 & 0.183 \\
& $\tau_b$ & 0.119 & 200 & 0.081 & \textbf{0.074} & 0.117 & \missingresult & \textbf{0.074} & 0.075 & 0.082 & \missingresult & \textbf{0.074} & 0.075 & 0.081 & \missingresult \\
& & & 400 & 0.039 & \textbf{0.022} & 0.055 & \missingresult & 0.049 & \textbf{0.047} & 0.051 & \missingresult & 0.048 & \textbf{0.047} & 0.050 & \missingresult \\
\midrule
\multirow{10}{*}{NP$^\ast$}
& $\beta_1^{\mathrm{cure}}$ & -1.000 & 200 & 0.145 & 0.150 & 0.242 & \textbf{0.004} & 0.737 & 0.802 & 1.298 & \textbf{0.392} & 0.723 & 0.788 & 1.275 & \textbf{0.392} \\
& & & 400 & 0.103 & 0.100 & 0.142 & \textbf{0.013} & 0.311 & 0.314 & 0.530 & \textbf{0.293} & \textbf{0.293} & 0.298 & 0.510 & \textbf{0.293} \\
& $\beta_2^{\mathrm{cure}}$ & 1.000 & 200 & 0.021 & \textbf{0.008} & 0.071 & 0.009 & 0.320 & \textbf{0.316} & 0.349 & \textbf{0.316} & 0.320 & \textbf{0.316} & 0.342 & \textbf{0.316} \\
& & & 400 & 0.022 & \textbf{0.011} & 0.069 & \textbf{0.011} & 0.235 & \textbf{0.232} & 0.255 & 0.233 & 0.234 & \textbf{0.232} & 0.246 & 0.233 \\
& $\beta_1^{\mathrm{surv}}$ & -1.000 & 200 & 0.178 & 0.088 & 0.382 & \textbf{0.003} & 0.383 & \textbf{0.319} & 0.563 & 0.351 & 0.340 & \textbf{0.306} & 0.414 & 0.352 \\
& & & 400 & 0.150 & 0.062 & 0.353 & \textbf{0.001} & 0.275 & \textbf{0.221} & 0.451 & 0.262 & 0.231 & \textbf{0.212} & 0.280 & 0.262 \\
& $\beta_2^{\mathrm{surv}}$ & 1.000 & 200 & 0.129 & \textbf{0.019} & 0.344 & 0.028 & 0.293 & \textbf{0.207} & 0.474 & 0.221 & 0.264 & \textbf{0.206} & 0.327 & 0.220 \\
& & & 400 & 0.117 & \textbf{0.008} & 0.340 & 0.013 & 0.213 & \textbf{0.139} & 0.408 & 0.151 & 0.178 & \textbf{0.139} & 0.227 & 0.151 \\
& $\tau_b$ & 0.155 & 200 & -0.006 & \textbf{-0.002} & -0.108 & \missingresult & \textbf{0.067} & \textbf{0.067} & 0.070 & \missingresult & \textbf{0.067} & \textbf{0.067} & 0.068 & \missingresult \\
& & & 400 & -0.009 & \textbf{-0.007} & -0.110 & \missingresult & \textbf{0.046} & \textbf{0.046} & 0.049 & \missingresult & \textbf{0.046} & \textbf{0.046} & \textbf{0.046} & \missingresult \\
\bottomrule
\end{tabular}

\begin{tablenotes}[flushleft]
\scriptsize
\item[]
\textit{Note:} $\ast$: Data generation for NP follows the simulation settings of \citet{niuSemiparametricMarginalMixture2013,niuMarginalRegressionAnalysis2014}, using a Clayton copula with parameter $2$. $\dagger$: The values for NP are based only on replications in which the fitting procedure converged and should therefore be regarded as reference values. \textbf{Boldface} indicates the smallest absolute RB or the smallest RMSE or SD among all four estimators within each row.
\end{tablenotes}
\end{sidewaystable}

When data are generated from the proposed model, the proposed estimator exhibits the smallest biases, particularly for the survival regression coefficients and Kendall's $\tau_b$. The proposed estimator also achieves the smallest RMSEs for many estimands. The CS estimator also performs reasonably well but shows larger biases for the survival regression coefficients and Kendall's $\tau_b$. The WLY estimator generally exhibits the largest biases and RMSEs. Although the NP estimator yields relatively small biases for the cure regression coefficients, its survival regression coefficients remain appreciably biased even among converged replications.

When data are generated from the CS model, the proposed and CS estimators yield similar results, with the CS estimator tending to have slightly smaller biases and RMSEs. This is consistent with the fact that the proposed model contains the CS model as the limiting case $\gamma \to 0$. Indeed, the proposed model gives estimates for gamma close to zero. The WLY estimator exhibits obvious biases and larger RMSEs, particularly for the survival regression coefficients and Kendall's $\tau_b$. The NP estimator exhibits small biases; however, its RMSEs are somewhat larger than those of the proposed and CS estimators, particularly for $\beta_1^{\mathrm{surv}}$.

When data are generated from the WLY model, the proposed and CS models are misspecified, whereas the WLY model is correctly specified. Nevertheless, the proposed estimator exhibits the smallest biases for both survival regression coefficients at both sample sizes, with RMSEs comparable to or smaller than those of the CS and WLY estimators. For the cure regression coefficients, the CS estimator generally yields smaller biases and RMSEs than the proposed and WLY estimators. The NP estimator exhibits appreciable biases for the survival regression coefficients, although its RMSEs are competitive among converged replications. Overall, the proposed estimator demonstrates robustness to model misspecification, particularly in estimating the survival regression coefficients.

When data are generated from the NP model, the proposed and CS estimators perform similarly for the cure regression coefficients, whereas the CS estimator has smaller biases and RMSEs for the survival regression coefficients. The WLY estimator generally exhibits the largest biases and RMSEs. Among converged NP fits, the NP estimator has the smallest or nearly smallest absolute biases and competitive RMSEs for the regression coefficients. Because the NP method does not provide an estimate of Kendall's $\tau_b$, only the other three estimators are compared for this measure; the proposed and CS estimators yield nearly identical results and smaller absolute biases than the WLY estimator.

In conclusion, the proposed estimator performs well under the proposed and CS data-generating models and remains competitive for the survival regression coefficients under the WLY model despite misspecification. Under the NP model, the CS estimator outperforms the proposed estimator for the survival regression coefficients, whereas the two perform similarly for the cure regression coefficients and Kendall's $\tau_b$. Overall, the proposed estimator gives accurate and robust estimates across all the data-generating scenarios.

\section{Real data analysis}
\label{sec:realdata}
In this section, we apply the proposed cure frailty-copula models to a real dataset and compare them with the CS, WLY, and RAC models, where RAC refers to the model proposed by \citet{rouzbahaniNewBivariateSurvival2025}. \Cref{subsec:real-data-setup} describes the dataset and analysis setup. In \Cref{subsec:model-selection}, we compare these models using likelihood-based model selection approaches, including the Akaike information criterion (AIC), Bayesian information criterion (BIC), and the LRT. The NP model is not included in the comparison because its estimation procedure is not likelihood-based. For the selected model, we interpret the parameter estimates and present the model-based marginal survival curves overlaid on the Kaplan--Meier (KM) curves in \Cref{subsec:selected-model}.

\subsection{Data and analysis setup}\label{subsec:real-data-setup}
We analyze the diabetic retinopathy dataset from the \textsf{R} package \texttt{survival} \citep[version 3.8-3]{survival-package}. The dataset originates from the Diabetic Retinopathy Study \citep{husterModellingPairedSurvival1989} and contains paired survival times for both eyes of 197 patients with diabetic retinopathy. For each patient, one eye was randomized to receive laser treatment and the other eye received no treatment. The endpoint of interest is the time from treatment initiation to vision loss. We index the treated and untreated eyes by $j = 1$ and $j = 2$, respectively, and denote the corresponding survival times by $T_1$ and $T_2$. Censoring was primarily patient-level (death, dropout, or study termination), so censoring times were common to both eyes. The sole exception occurred when one eye had already experienced the event; in that case, only the remaining eye could be censored. The covariates considered in the analysis are age (in years, standardized), denoted by $x_{i,\text{age}}$, and risk score (6--12), denoted by $x_{ij,\text{score}}$. Age is a patient-level covariate and thus does not depend on the margin $j$, whereas the risk score is computed separately for each eye and hence may vary with margin $j$. We incorporate the covariates into both the marginal cure fractions and the susceptible survival-time distributions. Accordingly, we extend the competing models to include covariate effects on both components.

\subsection{Model selection results}\label{subsec:model-selection}
We summarize the model comparison in \Cref{tab:model_selection}. For the proposed and CS models, we consider the independence and Gumbel copulas and three specifications for the odds ratio $\OR$ governing dependence between the cure indicators: (i) $0 < \OR < 1$, representing negative association; (ii) $\OR = 1$, indicating no association; and (iii) $\OR > 1$, representing positive association. For the WLY model, we also consider the three specifications. The RAC model is fitted with the independence and Gumbel copulas. Under this model, the cure-status odds ratio is induced by the other model parameters.

\begin{table}[htbp]
  \centering
  \caption{Model selection for the real data analysis with covariates}
  \label{tab:model_selection}
  \renewcommand{\arraystretch}{1}
  \begin{tabular*}{\textwidth}{@{\extracolsep{\fill}}lllcccc@{}}
    \toprule
    Model & Copula & Odds ratio & AIC & BIC & Log-likelihood & \#par \\
    \midrule
    \multirow[t]{8}{*}{Proposed}
      & \multirow[t]{4}{*}{Gumbel} & $\OR = 1$              & 1656.117 & 1708.648 & $-812.059$ & 16 \\
      & & $0 < \OR < 1$                 & 1655.120 & 1710.935 & $-810.560$ & 17 \\
      & & $1 < \OR < \infty$            & 1658.117 & 1713.932 & $-812.059$ & 17 \\
      \addlinespace[0.2em]
      & \multirow[t]{4}{*}{Independence} & $\OR = 1$ & 1655.011 & \textbf{1704.259} & $-812.506$ & 15 \\
      & & $0 < \OR < 1$                 & \textbf{1653.953} & 1706.484 & $-810.976$ & 16 \\
      & & $1 < \OR < \infty$            & 1657.016 & 1709.547 & $-812.508$ & 16 \\
    \midrule
    \multirow[t]{8}{*}{CS}
      & \multirow[t]{4}{*}{Gumbel} & $\OR = 1$              & 1664.418 & 1713.666 & $-817.209$ & 15 \\
      & & $0 < \OR < 1$                 & 1665.190 & 1717.721 & $-816.595$ & 16 \\
      & & $1 < \OR < \infty$            & 1667.909 & 1720.440 & $-817.955$ & 16 \\
      \addlinespace[0.2em]
      & \multirow[t]{4}{*}{Independence} & $\OR = 1$              & 1678.572 & 1724.537 & $-825.286$ & 14 \\
      & & $0 < \OR < 1$                 & 1680.351 & 1729.599 & $-825.175$ & 15 \\
      & & $1 < \OR < \infty$            & 1666.914 & 1716.162 & $-818.457$ & 15 \\
    \midrule
    \multirow[t]{4}{*}{WLY}
      & \multirow[t]{4}{*}{--} & $\OR = 1$ & 1657.011 & 1709.542 & $-812.505$ & 16 \\
      & & $0 < \OR < 1$                 & 1655.953 & 1711.767 & $-810.976$ & 17 \\
      & & $1 < \OR < \infty$            & 1659.019 & 1714.834 & $-812.510$ & 17 \\
    \midrule
    \multirow[t]{2}{*}{RAC}
      & Gumbel & $\OR>1$\textsuperscript{*} & 1657.914 & 1710.446 & $-812.957$ & 16 \\
      & Independence & $\OR>1$\textsuperscript{*} & 1656.928 & 1706.176 & $-813.464$ & 15 \\
    \bottomrule
  \end{tabular*}
  \begin{minipage}{\textwidth}
    Note: \#par denotes the number of parameters. \textbf{Boldface} indicates the smallest AIC and BIC. \textsuperscript{*} For the RAC model, $\OR>1$ is not a free parameter but is determined by the other model parameters.
  \end{minipage}
\end{table}
\vspace{-2em}
As shown in \Cref{tab:model_selection}, the proposed model with the independence copula attains the best AIC (with $0<\OR<1$) and the best BIC (with $\OR=1$). The proposed model with the Gumbel copula also attains the largest maximized log-likelihood (with $0<\OR<1$). Taken together, the better fit of the proposed models relative to the CS, WLY, and RAC models suggests the practical value of the proposed framework. The proposed model accommodates frailty, whereas the CS model does not, and combines a shared frailty with a separately specified copula, whereas the WLY model relies on correlated frailties. It also explicitly models dependence between cure indicators within a mixture cure formulation, whereas the RAC model uses a non-mixture cure formulation in which cure dependence is determined implicitly by other parameters.

We further assess $\OR=1$ using the likelihood ratio test described in \Cref{subsec:lrt}. Specifically, we test $H_0:\OR=1$ against $H_1:\OR\neq 1$ for the selected independence copula model by AIC and BIC. As shown in \Cref{tab:model_selection}, we obtain the likelihood ratio statistic $\lambda = 2\left\{-810.976-(-812.506)\right\} = 3.060$. Since $\lambda=3.060$ is smaller than $\chi^2_{1,0.05}=3.841$, we do not reject $H_0$. Clinically, the failure to reject $\OR=1$ suggests that laser treatment does not directly influence the cure outcome of the untreated eye. Such an assessment is not possible under the RAC model, in which the odds ratio is determined by the other parameters and restricted to $\OR>1$; by contrast, the odds ratio is a free parameter in the proposed model, permitting a direct test of independence. Given the moderate sample size ($n=197$), the LRT may have limited power to detect weak cure dependence (see \Cref{subsec:sim-lrt}). While the failure to reject $H_0$ suggests that strong cure dependence is unlikely, weaker dependence may have gone undetected.

Although the model with $0<\OR<1$ attained the smallest AIC, its improvement over the model with $\OR=1$ was modest. Moreover, BIC favored $\OR=1$, and the LRT did not reject cure-status independence. We therefore selected the more parsimonious model with $\OR=1$ for ease of interpretation.

\subsection{Results under the selected model}\label{subsec:selected-model}
We now examine the fitted results under the best model (the proposed model with independence copula and $\OR=1$). In particular, we report the parameter estimates and assess the implied cure fractions, dependence structure, and adequacy of the fitted marginal survival functions.

\begin{table}[htbp]
  \caption{Maximum likelihood estimates for the zero-inflated gamma cure frailty independence copula model with $\OR = 1$ fitted to the diabetic retinopathy dataset}
  \label{tab:est_params_retino_cov}
  \centering
  \scriptsize
  \renewcommand{\arraystretch}{1.20}
  \begin{tabular}{c S S c}
    \toprule
    {Parameter} & {Estimate} & {SE} & {95\% CI} \\
    \midrule
    $\gamma$ & 1.852 & 0.570 & $[1.013,\, 3.386]$ \\
    $a_1$ & 1.222 & 0.216 & $[0.865,\, 1.727]$ \\
    $r_1$ & 0.012 & 0.006 & $[0.005,\, 0.033]$ \\
    $a_2$ & 1.236 & 0.155 & $[0.966,\, 1.581]$ \\
    $r_2$ & 0.017 & 0.006 & $[0.008,\, 0.035]$ \\
    $\beta_{01}^{\mathrm{cure}}$ & -0.688 & 0.853 & $[-2.360,\, 0.983]$ \\
    $\beta_{02}^{\mathrm{cure}}$ & -3.898 & 3.562 & $[-10.880,\, 3.084]$ \\
    $\beta_{1,\text{age}}^{\mathrm{cure}}$ & -0.670 & 0.822 & $[-2.280,\, 0.941]$ \\
    $\beta_{1,\text{score}}^{\mathrm{cure}}$ & -0.819 & 0.518 & $[-1.833,\, 0.196]$ \\
    $\beta_{2,\text{age}}^{\mathrm{cure}}$ & -1.194 & 0.853 & $[-2.865,\, 0.477]$ \\
    $\beta_{2,\text{score}}^{\mathrm{cure}}$ & 1.496 & 2.658 & $[-3.714,\, 6.705]$ \\
    $\beta_{1,\text{age}}^{\mathrm{surv}}$ & -0.725 & 0.291 & $[-1.294,\, -0.155]$ \\
    $\beta_{1,\text{score}}^{\mathrm{surv}}$ & -0.350 & 0.288 & $[-0.914,\, 0.214]$ \\
    $\beta_{2,\text{age}}^{\mathrm{surv}}$ & 0.204 & 0.193 & $[-0.175,\, 0.582]$ \\
    $\beta_{2,\text{score}}^{\mathrm{surv}}$ & 0.612 & 0.188 & $[0.243,\, 0.981]$ \\
    \bottomrule
  \end{tabular}
\end{table}

\Cref{tab:est_params_retino_cov} presents the parameter estimates, SEs, and $95\%$ CIs for all model parameters. In the cure component, neither age nor risk score was significantly associated with the cure probability in either eye, as all corresponding $95\%$ CIs included zero. In the susceptible survival-time component, age was associated with a lower hazard of vision loss in the treated eye ($\hat{\beta}_{1,\mathrm{age}}^{\mathrm{surv}}=-0.725$, $95\%$ CI: $[-1.294,-0.155]$), whereas a higher risk score was associated with a higher hazard in the untreated eye ($\hat{\beta}_{2,\mathrm{score}}^{\mathrm{surv}}=0.612$, $95\%$ CI: $[0.243,0.981]$). Neither risk score in the treated eye nor age in the untreated eye was significantly associated with the hazard of vision loss.

Although survival times are assumed to be conditionally independent given frailty through the independence copula, the estimated frailty parameter $\hat{\gamma} = 1.852$ yields a positive Kendall's tau for the uncured subpopulation ($\tau_{00} = 0.481$). By contrast, the overall rank correlation coefficients accounting for cure are approximately $\tau_b = 0.182$ and $\rho_b = 0.238$, implying weak population-level dependence. This discrepancy reflects the attenuation of association by the presence of cured subjects. Overall, the dependence structure is primarily explained by shared frailty among uncured subjects, whereas the cure process itself appears to operate independently across eyes.

\begin{figure}[htbp]
    \begin{center}
            \includegraphics[width=0.95\textwidth]{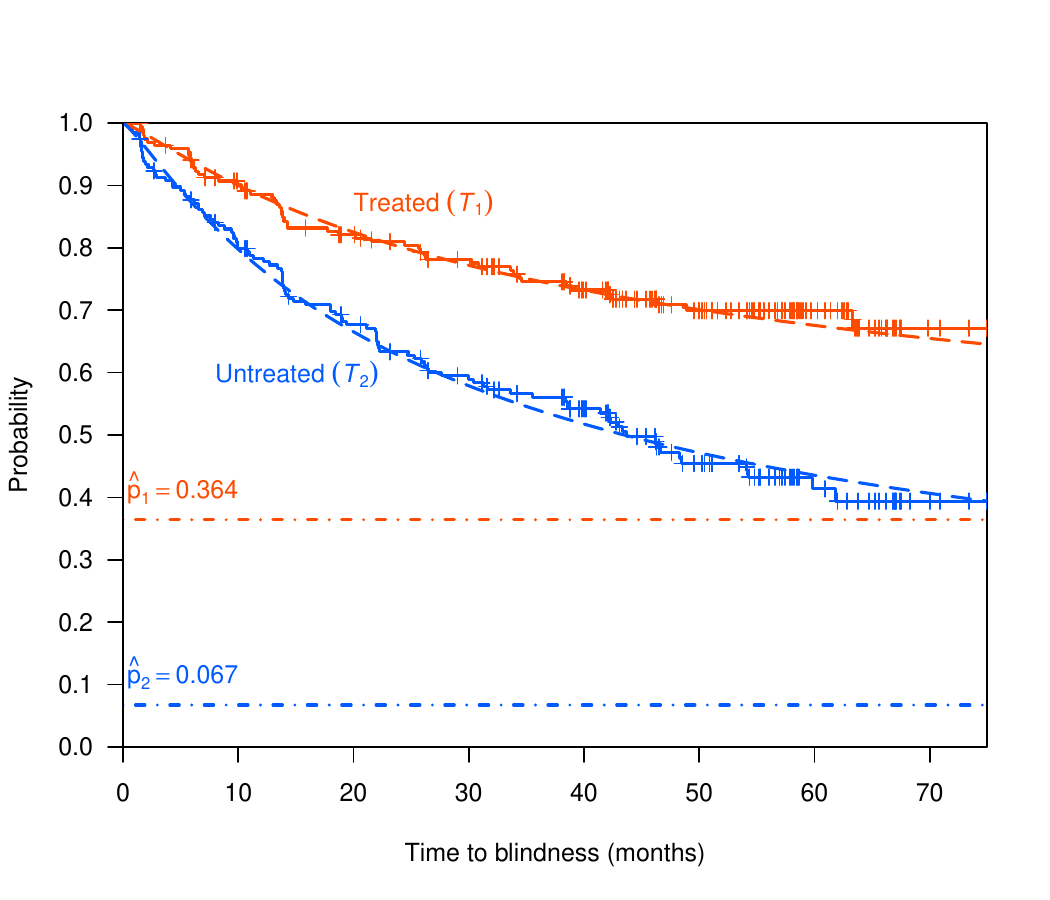}
    \end{center}
    \setlength{\abovecaptionskip}{-1em}
    \caption{KM estimators (solid lines), model-based marginal survival curves (dashed lines), and estimated marginal cure fractions (horizontal dot-dashed lines) for $T_1$ (treated eye, red) and $T_2$ (untreated eye, blue) in the diabetic retinopathy dataset}
    \label{fig:km_vs_fitted}
\end{figure}

To further assess model adequacy, we compare the model-based marginal survival functions with the KM estimators. The marginal cure fractions are computed by averaging over covariates as
\begin{align*}
  p_j
  &\coloneqq \E\left[p_j\left(\boldsymbol{x}_{i j}\right)\right]
  = \frac{1}{n}\sum_{i=1}^{n}
  \frac{\exp\left(\beta_{0j}^{\mathrm{cure}} + \beta_{j,\text{age}}^{\mathrm{cure}} x_{i,\text{age}} + \beta_{j,\text{score}}^{\mathrm{cure}} x_{ij,\text{score}}\right)}
  {1 + \exp\left(\beta_{0j}^{\mathrm{cure}} + \beta_{j,\text{age}}^{\mathrm{cure}} x_{i,\text{age}} + \beta_{j,\text{score}}^{\mathrm{cure}} x_{ij,\text{score}}\right)}, \quad j = 1, 2,
\end{align*}
where $n = 197$ is the total number of patients.

\Cref{fig:km_vs_fitted} overlays the fitted marginal survival functions (dashed lines) and the KM estimators (solid lines) for both eyes. The fitted curves closely match the KM estimators, confirming that the model adequately captures the overall survival behavior. In particular, the treated eye maintains a higher survival probability throughout the follow-up period, suggesting that laser treatment effectively delays vision loss. However, under right censoring, the KM estimator cannot directly recover tail behavior and can at most provide an upper bound on the cure fraction. In contrast, the proposed cure model offers a principled framework for estimating the cure fraction: by plugging the estimates into the marginal cure fraction formula, we obtain $\hat{p}_1 = 0.364$ and $\hat{p}_2 = 0.067$, shown as horizontal dot-dashed lines in the figure. These estimates characterize the long-term survival pattern induced by the presence of cured subjects.

\section{Conclusion}\label{sec:conclusion}
In this paper, we proposed a broad class of bivariate cure frailty-copula models for paired survival data that includes several existing models as special cases. The proposed class accommodates cure fractions, continuous frailty-induced heterogeneity, and flexible copula-based dependence. The model separates dependence into three components: cure-status dependence via a common odds ratio, shared gamma frailty among susceptible subjects, and residual copula dependence between susceptible survival times. We also derived explicit expressions for tie-adjusted Kendall's tau and Spearman's rho in the presence of a cure fraction, thereby providing interpretable rank-based measures of dependence for bivariate cure models both at the population level and within the uncured subpopulation. To our knowledge, such rank-based dependence measures have not previously been explicitly formulated in the setting of bivariate cure models.

For suitable choices of baseline marginal survival functions and copulas, we derived closed-form expressions for the joint survival function after integrating out frailty, enabling computationally tractable likelihood-based inference. We developed maximum likelihood estimation and a likelihood ratio test of $H_0:\OR=1$, and verified their sound finite-sample performance through simulation studies. In the Diabetic Retinopathy Study analysis, the proposed model achieved better overall fit than existing approaches. It also revealed dependence patterns that differ from those suggested by existing models. These findings highlight the benefit of separately modeling cure-status dependence, frailty-induced dependence, and residual dependence among susceptible survival times.

Several directions merit further investigation. Although $\OR=1$ was selected for the dataset in \Cref{sec:realdata}, settings with $\OR>1$ are clinically plausible (e.g., bilateral kidney infections after antimicrobial therapy). The proposed framework is directly applicable to such datasets and provides a practical tool for researchers to model dependence in cure fractions as well as in survival outcomes. Beyond such applications, methodological extensions are also of interest. First, motivated by more flexible frailty distributions—such as the power variance function (PVF) family—and mixture frailty models, it would be of interest to extend the proposed framework beyond gamma frailty. A key question is when the bivariate survival function remains available in closed form after marginalizing over frailty, depending on the chosen marginals and copula. Second, while a lower-tail dependence coefficient can be defined naturally within the susceptible subpopulation, defining an upper-tail dependence coefficient at the population level in the presence of a cure fraction warrants further study.

\backmatter

\bmhead*{Supplementary information}

Online Resource 1 provides further details on the proposed model and its relationships with existing models (Section S1); derivations of joint cure-status probabilities under a fixed odds ratio (Section S2); derivations, additional properties, and numerical validation of the rank correlation coefficients (Sections S3--S4); technical details for likelihood-based inference (Section S5); additional simulation settings and results (Section S6); and a comparison of the EM algorithm with direct optimization (Section S7).

\bmhead*{Acknowledgements}

The authors thank the editor, the associate editor, and an anonymous reviewer for their valuable comments and suggestions, which substantially improved the manuscript.

This work was supported in part by The Graduate University for Advanced Studies, SOKENDAI. The second author was supported financially by JSPS KAKENHI (24K00273, 25K15037), and the third author by JSPS KAKENHI (22K11948). This study was carried out under the ISM Cooperative Research Program (2026-ISMCRP-5004).

\bmhead*{Code availability}

The \textsf{R} package \texttt{curecopula} will be submitted to CRAN and is also available on GitHub. The code used to reproduce the figures and tables in this paper is maintained in a separate GitHub repository. Both repositories are currently private and will be made publicly available later in the review process. The results reported in this paper were generated using \texttt{curecopula} version 0.1.0.

\bibliography{ref}

\end{document}


\title{Supplementary Materials for \\ ``A bivariate cure copula model with zero-inflated gamma frailty for right-censored data''}

\author[1]{\fnm{Masaki} \sur{Hino}\orcidicon{0009-0006-7734-6335}}

\author[1,2]{\fnm{Shogo} \sur{Kato}\orcidicon{0000-0003-2648-6769}}

\author*[3]{\fnm{Takeshi} \sur{Emura}\orcidicon{0000-0002-3904-4014}}\email{takeshiemura@gmail.com}

\affil[1]{\orgname{Graduate Institute for Advanced Studies, SOKENDAI}, \orgaddress{\street{10-3 Midori-cho}, \city{Tachikawa}, \postcode{190-8562}, \state{Tokyo}, \country{Japan}}}
 
\affil[2]{\orgname{Institute of Statistical Mathematics}, \orgaddress{\street{10-3 Midori-cho}, \city{Tachikawa}, \postcode{190-8562}, \state{Tokyo}, \country{Japan}}}
 
\affil*[3]{\orgdiv{School of Informatics and Data Science}, \orgname{Hiroshima University}, \orgaddress{\street{1-3-2 Kagamiyama}, \city{Higashi-Hiroshima}, \postcode{739-8511}, \state{Hiroshima}, \country{Japan}}}

\abstract{This Online Resource accompanies the article ``A bivariate cure copula model with zero-inflated gamma frailty for right-censored data''. It contains additional details on the proposed model, including a derivation of the bivariate survival function, the closed-form expression for the Farlie--Gumbel--Morgenstern (FGM) copula model, and relationships with existing models (Section S1); derivations of joint cure-status probabilities under a fixed odds ratio (Section S2); derivations and additional properties of population rank correlation coefficients (Section S3); a numerical validation of the rank correlation coefficients (Section S4); technical details for likelihood-based inference (Section S5); the data-generation algorithm for the simulation with covariates and simulation results for the model without covariates (Section S6); and a comparison of the EM algorithm with direct optimization (Section S7).}

\maketitle

\noindent \textbf{Journal:} Computational Statistics

\clearpage
\section{Additional details on the proposed model}
\label{app:surv}
\subsection{Derivation of the bivariate survival function}
Let $T_j$ and $Y_j$ denote the survival time and cure indicator for margin $j$, respectively, with marginal cure probability $p_j=\Prob(Y_j=1)$. We write $p_{rs}=\Prob(Y_1=r,Y_2=s)$ for $r,s\in\{0,1\}$ and let $C_\theta$ denote a copula with dependence parameter $\theta$.

The shared frailty $W$ follows a gamma distribution with shape parameter $1/\gamma$ and scale parameter $\gamma$. Its probability density function is
\[
  f_\gamma(w)
  = \frac{1}{\Gamma(1/\gamma)\,\gamma^{1/\gamma}}w^{1/\gamma-1}
  \exp\left(-\frac{w}{\gamma}\right),
  \quad w>0,\ \gamma>0,
\]
where $\Gamma(\cdot)$ denotes the gamma function.

In this section, we derive the unconditional bivariate survival function under the proposed model. The key idea is to apply the law of total expectation by conditioning on the latent cure indicators $(Y_1, Y_2)$ and the shared gamma frailty $W$. Recall that $Z_j=(1-Y_j)W$, where $Y_j=1$ indicates cure and $Y_j=0$ indicates susceptibility. Thus, $Z_j=0$ for cured subjects and $Z_j=W$ otherwise. Let $S_j(t_j\mid W)$ denote the conditional marginal survival function given $W$.
\begin{align*}
    S(t_1, t_2)
    &= \Prob(T_1 > t_1, T_2 > t_2) \\
    &= \E_{Z_1, Z_2} [\Prob(T_1 > t_1, T_2 > t_2 \mid Z_1, Z_2)] \\
    &= \E_{Y_1, Y_2, W} [\Prob(T_1 > t_1, T_2 > t_2 \mid Y_1, Y_2, W)] \\
    &= \Prob(Y_1 = 1, Y_2 = 1) \Prob(T_1 > t_1, T_2 > t_2 \mid Z_1 = 0, Z_2 = 0)  \\
    &\quad + \Prob(Y_1 = 0, Y_2 = 1) \E_W [\Prob(T_1 > t_1, T_2 > t_2 \mid Z_1 = W, Z_2 = 0)] \\
    &\quad + \Prob(Y_1 = 1, Y_2 = 0) \E_W [\Prob(T_1 > t_1, T_2 > t_2 \mid Z_1 = 0, Z_2 = W)] \\
    &\quad + \Prob(Y_1 = 0, Y_2 = 0) \E_W [\Prob(T_1 > t_1, T_2 > t_2 \mid Z_1 = W, Z_2 = W)] \\
    &= p_{11}
    + p_{01} \E_W [\Prob(T_1 > t_1 \mid Z_1 = W)]
    + p_{10} \E_W [\Prob(T_2 > t_2 \mid Z_2 = W)] \\
    &\quad + p_{00} \E_W [\Prob(T_1 > t_1, T_2 > t_2 \mid Z_1 = W, Z_2 = W)] \\
    &= p_{11} + p_{01} \E_W [S_1(t_1 \mid W)] + p_{10} \E_W[S_2(t_2 \mid W)] \\
    &\quad + p_{00} \E_W[C_{\theta} (S_1(t_1 \mid W), S_2(t_2 \mid W))]
\end{align*}
To express the conditional marginal survival functions in terms of the baseline survival functions, we use the multiplicative frailty specification \citep{wienkeFrailtyModelsSurvival2010}. Specifically, given the shared frailty $W$, the conditional marginal hazard of $T_j$ is
\[
  h_j(t \mid W) = W h_{0j}(t),
\]
where $h_{0j}(t)$ denotes the baseline hazard for $j=1,2$. By the standard relationship between the hazard and survival functions, we obtain
\[
  S_j(t_j \mid W)
  = \exp\left\{-\int_{0}^{t_j} h_j(u \mid W)\,\mathrm{d}u\right\}
  = S_{0j}(t_j)^W,
\]
where $S_{0j}(t_j)=\exp\left\{-\int_{0}^{t_j} h_{0j}(u)\,\mathrm{d}u\right\}$ is the baseline survival function.
\begin{align*}
    S(t_1, t_2)
    &= p_{11} + p_{01} \E_W [S_{01}(t_1)^W] + p_{10} \E_W [S_{02}(t_2)^W] \\
    &\quad + p_{00} \E_W [C_{\theta}(S_{01}(t_1)^W, S_{02}(t_2)^W)] \\
    &= p_{11} + p_{01} \int_0^{\infty} S_{01}(t_1)^w f_\gamma(w) \, \mathrm{d}w + p_{10} \int_0^{\infty} S_{02}(t_2)^w f_\gamma(w) \, \mathrm{d}w \\
    &\quad + p_{00} \int_0^{\infty} C_{\theta}\left(S_{01}(t_1)^w, S_{02}(t_2)^w\right) f_\gamma(w) \, \mathrm{d}w.
\end{align*}
The integrability of the last term depends on the choice of copula.

\subsection{Closed-form expression for the FGM copula model}
The main text presents closed-form expressions for the independence and Gumbel copula models. Here, we additionally provide the corresponding expression for the FGM copula model. The FGM copula is defined as
\[
  C_{\theta}(u, v)
  = uv \left\{1 + \theta (1 - u)(1 - v)\right\}, \quad \theta \in [-1, 1].
\]
Under the Weibull baseline survival functions $S_{0j}(t_j)=\exp\left(-r_jt_j^{a_j}\right)$ used in the main text, where $a_j$ and $r_j$ are the shape and scale parameters, respectively, the bivariate survival function is expressed as
\begin{align*}
  S(t_1, t_2)
  &= p_{11} + p_{01} \left(1 + \gamma r_1 t_1^{a_1} \right)^{-1/\gamma} + p_{10} \left(1 + \gamma r_2 t_2^{a_2} \right)^{-1/\gamma} \\
  &\quad + p_{00} \left[(1 + \theta)\left\{1 + \gamma\left(r_1 t_1^{a_1} + r_2 t_2^{a_2}\right)\right\}^{-1/\gamma} - \theta \left\{1 + \gamma\left(2 r_1 t_1^{a_1} + r_2 t_2^{a_2}\right)\right\}^{-1/\gamma} \right. \\
  &\qquad \left. - \theta \left\{1 + \gamma\left(r_1 t_1^{a_1} + 2 r_2 t_2^{a_2}\right)\right\}^{-1/\gamma} + \theta \left\{1 + 2 \gamma\left(r_1 t_1^{a_1} + r_2 t_2^{a_2}\right)\right\}^{-1/\gamma}\right].
\end{align*}
Setting $\theta = 0$ reduces to the independence copula model in Example 1 of Section 2.1 of the main text.

Let $S_{j}(t_j) = \left(1 + \gamma r_j t_j^{a_j} \right)^{-1/\gamma}$ for $j = 1, 2$. The fourth term can be written as the generalized FGM copula $C_{\theta,\gamma}$, defined by
\begin{align}
  C_{\theta, \gamma} (u, v)
  &= \left(1+\theta\right) \left(u^{-\gamma} + v^{-\gamma} - 1\right)^{-1/\gamma}
  - \theta \left(2 u^{-\gamma} + v^{-\gamma} - 2\right)^{-1/\gamma} \notag \\
  &\quad - \theta \left(u^{-\gamma} + 2v^{-\gamma} - 2\right)^{-1/\gamma}
  + \theta \left(2u^{-\gamma} + 2v^{-\gamma} - 3\right)^{-1/\gamma},
  \label{eq:one-gen-fgm}
\end{align}
where $\theta \in [-1,1]$ and $\gamma > 0$. The functional form of $C_{\theta, \gamma}$ was proposed by \citet{cookGeneralizedBurrParetoLogisticDistributions1986} (not explicitly in terms of copulas), and it was later considered as a copula by \citet{wangMultivariateFailureTime2021}. This copula belongs to the generalized FGM family and allows a wider attainable range of dependence, in terms of Kendall's tau, than the classical FGM copula; see Section S3.3.

\subsection{Relationships with existing models}
The proposed model contains several existing models as special or limiting cases. When $\gamma\to0$, the shared frailty degenerates to $W=1$, and the joint survival function in the main text reduces to
\[
  S(t_1,t_2)=p_{11}+p_{01}S_{01}(t_1)+p_{10}S_{02}(t_2)+p_{00}C_\theta\left\{S_{01}(t_1),S_{02}(t_2)\right\}.
\]
Thus, the proposed framework includes the cure-copula model of \citet{chatterjeeBivariateCureMixtureApproach2001} as a limiting case. It also yields the cure-copula formulation of \citet{kimCureRateModel2017} under their specification of the cure-status probabilities, which differs from the odds-ratio parameterization used in the main text.

The relationship with the model of \citet{wienkeBivariateFrailtyModel2003} is different. The model uses correlated gamma frailties rather than a shared frailty combined with a separately specified conditional copula. For susceptible subjects, the dependence component of the joint survival function is represented by the copula
\begin{equation}
  C_{\varrho,\sigma^2}(u,v)
  =u^{1-\varrho}v^{1-\varrho}\left(u^{-\sigma^2}+v^{-\sigma^2}-1\right)^{-\varrho/\sigma^2},
  \quad 0\leq\varrho\leq1,\ \sigma^2>0,
  \label{eq:wly-surv}
\end{equation}
where $\varrho$ is the correlation coefficient between the frailties and $\sigma^2$ is their variance parameter. Although \cref{eq:wly-surv} has a copula representation, its functional form is induced by the correlated gamma-frailty distribution and cannot be specified independently of that distribution. By contrast, the proposed model uses a shared frailty and an explicit conditional copula to represent residual dependence. Consequently, neither model generally contains the other. When $\varrho=1$, however, the two frailties coincide, and the susceptible survival-time component of the model reduces to the shared gamma-frailty component obtained under the independence copula in the proposed model.

When the Gumbel copula is used with gamma frailty in the model of \citet{rouzbahaniNewBivariateSurvival2025}, its cure probabilities are
\[
  p_{11}=\left\{1+\gamma \left(\lambda_1^{\theta+1}+\lambda_2^{\theta+1}\right)^{1/(\theta+1)}\right\}^{-1/\gamma},
  \qquad p_j=(1+\gamma\lambda_j)^{-1/\gamma},\quad j=1,2,
\]
where $\theta \geq 0$. Since $(\lambda_1^{\theta+1}+\lambda_2^{\theta+1})^{1/(\theta+1)}\leq\lambda_1+\lambda_2$, we have $p_{11}>p_1p_2$ for $\gamma,\lambda_1,\lambda_2>0$, and hence
\[
  \OR
  =\frac{p_{11}(1-p_1-p_2+p_{11})}{(p_1-p_{11})(p_2-p_{11})}>1.
\]
The independence copula is obtained by setting $\theta=0$ and the same inequality holds. Thus, the cure-status odds ratio of their model is fixed by the other model parameters rather than parameterized separately.

Finally, when $p_1\to0$ and $p_2\to0$, so that no cure fraction remains, the proposed model reduces to
\[
  S(t_1,t_2)=\int_0^\infty C_\theta\left\{S_{01}(t_1)^w,S_{02}(t_2)^w\right\}f_\gamma(w)\,\mathrm{d}w.
\]
This is the frailty-copula model of \citet{wangLikelihoodbasedInferenceFrailtycopula2020} and may also be viewed as the bivariate counterpart of the multivariate model studied by \citet{wangMultivariateFailureTime2021}.

\section{Derivation of the joint probabilities under a fixed odds ratio}
\label{app:OR-deriv}
Marginal cure rates are given by $p_j = \Prob(Y_j=1) \in (0, 1),\, j=1,2$. The joint probabilities are $p_{rs}=\Prob(Y_1=r, Y_2=s)$ for $r,s \in \{0,1\}$. The odds ratio is defined as $\OR = (p_{11}p_{00})/(p_{10}p_{01})$. Then the following system of equations is satisfied:
\begin{equation*}
  p_{11} + p_{10} = p_1,
  \qquad p_{11} + p_{01} = p_2, 
  \qquad p_{11} + p_{10} + p_{01} + p_{00} = 1.
\end{equation*}
Thus, we have $p_{10}=p_1-p_{11}, p_{01}=p_2-p_{11}, p_{00}=1-p_1-p_2+p_{11}$. Substituting these expressions for $p_{10}$, $p_{01}$, and $p_{00}$ into the definition of $\OR$, we obtain an equation in terms of $p_{11}$:
\begin{align*}
    &\OR = \frac{p_{11} p_{00}}{p_{10} p_{01}}
    = \frac{p_{11}(1+p_{11}-p_1-p_2)}{(p_1-p_{11})(p_2 - p_{11})} \\
    &\iff
    0 = (\OR - 1) p_{11}^2 - \{(\OR - 1)(p_1 + p_2) + 1\} p_{11} + \OR p_1 p_2.
\end{align*}
If $\OR = 1$, the joint probabilities are given by products of the marginals, and hence $Y_1$ and $Y_2$ are independent. If $\OR \neq 1$, then $Y_1$ and $Y_2$ are dependent. In this case, the above equation is quadratic in $p_{11}$. Let $D \coloneqq \{(\OR - 1)(p_1 + p_2) + 1\}^2 - 4 \OR (\OR - 1)p_1 p_2$ be the discriminant. Since it can be shown that $D > 0$, the equation has two distinct real roots. The solutions for $p_{11}$ are given by
\begin{equation*}
  p_{11} = \frac{(\OR - 1)(p_1 + p_2) + 1 \pm \sqrt{D}}{2(\OR - 1)}.
\end{equation*}
We denote the two solutions by $p_{11}^{+}$ and $p_{11}^{-}$ corresponding to the plus and minus signs, respectively.
We now show that $p_{11}^{-}$ satisfies the Fr\'echet bounds
\[
\max\{p_1+p_2-1, 0\} \leq p_{11}^{-} \leq \min\{p_1, p_2\},
\]
and hence yields a valid $2\times2$ joint distribution with marginals $(p_1,p_2)$, whereas $p_{11}^{+}$ never corresponds to a valid joint probability.

\medskip
\noindent\textbf{(I) $p_{11}^{-}$ satisfies the Fr\'echet bounds.}
The expression for $p_{11}^{-}$ coincides with the Plackett copula evaluated at $(p_1,p_2)$ with parameter $\OR$ \citep{nelsenIntroductionCopulas2006}. Since any copula $C$ satisfies the Fr\'echet--Hoeffding inequalities
\[
\max\{u+v-1, 0\} \leq C(u,v) \leq \min\{u,v\}, \quad u,v\in[0,1],
\]
we obtain, by setting $(u,v)=(p_1, p_2)$,
\[
\max\{p_1+p_2-1, 0\} \leq p_{11}^{-} \leq \min\{p_1, p_2\}.
\]
Hence $p_{10}=p_1-p_{11}^{-}\ge0$, $p_{01}=p_2-p_{11}^{-}\ge0$, and
$p_{00}=1-p_1-p_2+p_{11}^{-}\ge0$, so $p_{11}^{-}$ yields valid joint probabilities.

\medskip
\noindent\textbf{(II) $p_{11}^{+}$ is invalid.}
We consider two cases depending on the value of $\OR$.

\textit{Case (i): $0 < \OR < 1$.}
By direct computation, the product of the two roots satisfies
\[
  p_{11}^{-} p_{11}^{+} = \frac{\OR p_1 p_2}{\OR - 1}.
\]
Since $\OR - 1 < 0$ and $\OR p_1 p_2 > 0$, we have $p_{11}^{-} p_{11}^{+} < 0$. From part (I), $p_{11}^{-}\geq \max\{p_1+p_2-1, 0\}\geq 0$. Since $p_{11}^{-}p_{11}^{+}<0$, we have $p_{11}^{-}>0$, which implies $p_{11}^{+}<0$.

\textit{Case (ii): $\OR > 1$.}
Without loss of generality, assume $p_1 \leq p_2$. For $p_{11}$ to be a valid joint probability, it must satisfy $p_{11} \leq \min(p_1, p_2) = p_1$. We show that $p_{11}^{+} > p_1$, violating this constraint:
\begin{align*}
  p_{11}^{+} > p_1
  &\iff \frac{(\OR - 1)(p_1 + p_2) + 1 + \sqrt{D}}{2(\OR - 1)} > p_1 \\
  &\iff (\OR - 1)(p_1 + p_2) + 1 + \sqrt{D} > 2(\OR - 1)p_1 \\
  &\iff (\OR - 1)(p_2 - p_1) + 1 + \sqrt{D} > 0.
\end{align*}
The last inequality holds because $\OR > 1$, $p_2 \geq p_1$, and $\sqrt{D} > 0$. Therefore $p_{11}^{+} > p_1 = \min(p_1, p_2)$, which violates the Fr\'echet upper bound for joint probabilities.

Combining (I) and (II), we conclude that $\max\{p_1+p_2-1, 0\} \leq p_{11}^{-} \leq \min\{p_1, p_2\}$ and thus yields valid joint probabilities with marginals $(p_1,p_2)$, whereas $p_{11}^{+}$ never yields a valid joint probability.

\medskip
\noindent\textbf{Remark (continuity at $\OR=1$).}
Although the expression for $p_{11}^-(\OR)$ is not defined at $\OR=1$, it admits a continuous extension. By l'H\^opital's rule, $\lim_{\OR\to1} p_{11}^-(\OR)=p_1p_2$, so defining $p_{11}=p_1p_2$ at $\OR=1$ yields a continuous extension of $p_{11}^-(\OR)$. Indeed, letting $f(\OR)=(\OR-1)(p_1+p_2)+1$ and $D(\OR)=f(\OR)^2-4\OR(\OR-1)p_1p_2$, we can write
\[
p_{11}^-(\OR)=\frac{f(\OR)-\sqrt{D(\OR)}}{2(\OR-1)}.
\]
Applying l'H\^opital's rule gives
\[
\lim_{\OR\to1}p_{11}^-(\OR)=\frac{1}{2} \left(f'(1) - \frac{D'(1)}{2\sqrt{D(1)}}\right)=p_1p_2.
\]

\section{Population rank correlation coefficients: derivations and properties}
\label{app:rank_proof}
\citet{pimentelKendallsTauSpearmans2009} proposed Kendall's tau and Spearman's rho for bivariate zero-inflated random variables $(X, Y)$. We adapt their results to our cure-model setting, where a cured margin corresponds to an infinite event time. They consider four regions: $\{X=0, Y=0\}, \{X>0, Y=0\}, \{X=0, Y>0\}$, and $\{X>0, Y>0\}$, which correspond in our setting to $\{T_1 = \infty, T_2 = \infty\}$, $\{T_1 < \infty, T_2 = \infty\}$, $\{T_1 = \infty, T_2 < \infty\}$, and $\{T_1 < \infty, T_2 < \infty\}$, respectively. Because rank-based coefficients depend only on pairwise orderings and ties, the same decomposition applies.

\subsection{Proof of the result for Kendall's tau in Section 3.1 of the main text}
\label{app:proof-kendall}

\begin{proof}
    We first relate $\tau_b$ to the usual Kendall’s tau without tie correction. Consider two independent subjects $k$ and $l$ with observations
    \[
      (\boldsymbol{T}_k, \boldsymbol{T}_l) = ((T_{k1}, T_{k2}), (T_{l1}, T_{l2})).
    \]
    When there are no ties, Kendall’s tau is denoted by $\tau_a$; the tie-adjusted version is denoted by $\tau_b$ \citep{adlerModificationKendallsTau1957,urbanoTreatmentTiesAP2017}. In a sample $\{(t_{i1}, t_{i2})\}_{i=1}^n$, these are
    \[
      \tau_a = \frac{\sum_{i<j} \operatorname{sign}(t_{i1}, t_{j1})\cdot\operatorname{sign}(t_{i2}, t_{j2})}{n(n-1)/2},
      \qquad
      \tau_b = \frac{\sum_{i<j} \operatorname{sign}(t_{i1}, t_{j1})\cdot\operatorname{sign}(t_{i2}, t_{j2})}{\sqrt{n(n-1)/2 - t_{T_1}}\sqrt{n(n-1)/2 - t_{T_2}}},
    \]
    where $\operatorname{sign}(x, y)$ is the comparison sign defined by $\operatorname{sign}(x, y) = 1$ if $x>y$, $-1$ if $x<y$, and $0$ if $x=y$. This definition is well-defined on $(0, \infty]$ and avoids subtraction. In particular, it prevents an undefined expression such as $\infty - \infty$. The quantities $t_{T_1}$ and $t_{T_2}$ are the numbers of tied pairs in $T_1$ and $T_2$, respectively. In our setting, ties arise because cured subjects have $T=\infty$. As $n\to\infty$, we have
    \begin{equation}
        \tau_b
        = \tau_a \frac{n(n-1)/2}{\sqrt{n(n-1)/2 - t_{T_1}}\sqrt{n(n-1)/2 - t_{T_2}}}
        \ \longrightarrow \
        \frac{\tau_a}{\sqrt{(1 - p_1^2)(1 - p_2^2)}},
        \label{eq:tau_b}
    \end{equation}
    where $p_1$ and $p_2$ are the marginal cure fractions.
    
    Let $\mathcal{C}_\tau$ and $\mathcal{D}_\tau$ denote the concordant and discordant events, respectively:
    \begin{align*}
      \mathcal{C}_\tau
      &\coloneqq \{T_{k1} > T_{l1}, T_{k2} > T_{l2}\} \cup \{T_{k1} < T_{l1}, T_{k2} < T_{l2}\}, \\
      \mathcal{D}_\tau
      &\coloneqq \{T_{k1} > T_{l1}, T_{k2} < T_{l2}\} \cup \{T_{k1} < T_{l1}, T_{k2} > T_{l2}\}.
    \end{align*}
    Then
    \[ \tau_a = \Prob(\mathcal{C}_\tau) - \Prob(\mathcal{D}_\tau). \]

    Define the following regions: $\text{A} = \{T_1 = \infty,\, T_2 = \infty\}, \text{B} = \{T_1 < \infty,\, T_2 = \infty\}, \text{C} = \{T_1 = \infty,\, T_2 < \infty\}, \text{D} = \{T_1 < \infty,\, T_2 < \infty\}$. There are ten unordered pairs of regions. Pairs involving ties in at least one margin (A--A, A--B, A--C, B--B, C--C) contribute zero because $\operatorname{sign}(T_{kj}, T_{lj})=0$ for at least one $j\in\{1,2\}$.

    For the B--D pair, we always have $T_{k2} > T_{l2}$ since $T_{k2}=\infty$ and $T_{l2}<\infty$. Hence concordance/discordance is determined solely by the ordering of $T_{k1}$ and $T_{l1}$, and
    \[
        \Prob(T_{k1} < T_{l1} \mid \boldsymbol{T}_k \in \text{B}, \boldsymbol{T}_l \in \text{D})
        = \Prob(T_{k1} > T_{l1} \mid \boldsymbol{T}_k \in \text{B}, \boldsymbol{T}_l \in \text{D})
        = \frac{1}{2},
    \]
    because conditional on being finite, the marginal distribution of $T_1$ is continuous. Therefore the net contribution is zero. The same argument applies to the C--D pair.

    All A--D pairs are concordant, and all B--C pairs are discordant. Finally, for D--D pairs the contribution equals Kendall's tau within the uncured subpopulation, denoted by $\tau_{00}$. Combining these contributions yields
    \[
        \tau_a
        = 2 \left(p_{11} p_{00} - p_{01} p_{10}\right) + p_{00}^2 \tau_{00}.
    \]
    Together with \cref{eq:tau_b}, this proves the formula for $\tau_b$. It remains to express $\tau_{00}$ in terms of $C_{\theta,\gamma}^\ast(u,v)$. Let $\mathcal{L}_W(s)=\E[\exp(-sW)]$ denote the Laplace transform of $W$. Then $S_j(t_j) = \mathcal{L}_W\left(-\log S_{0j}(t_j)\right) = \left(1-\gamma \log S_{0j}(t_j)\right)^{-1/\gamma}$ is the marginal survival function of the uncured subpopulation for margin $j$. Setting $u = S_1(t_1)$ and $v = S_2(t_2)$, we have $S_{01}(t_1) = \exp\left\{-\left(u^{-\gamma}-1\right)/\gamma\right\}$ and $S_{02}(t_2) = \exp\{-(v^{-\gamma}-1)/\gamma\}$. The joint survival function of the uncured subpopulation can then be written as
    \[
      C_{\theta,\gamma}^\ast(u, v)
      = \int_{0}^{\infty} C_\theta \left(\exp\left\{-w \frac{u^{-\gamma} - 1}{\gamma}\right\}, \exp\left\{-w \frac{v^{-\gamma} - 1}{\gamma}\right\}\right) f_\gamma(w) \, \mathrm{d} w.
    \]
    Since $C_{\theta,\gamma}^\ast(u,1) = \mathcal{L}_W\left((u^{-\gamma}-1)/\gamma\right) = u$ and $C_{\theta,\gamma}^\ast(1,v) = v$, the function $C_{\theta,\gamma}^\ast$ satisfies the boundary conditions of a copula. The $2$-increasing property follows from the fact that $C_\theta$ is $2$-increasing and integration preserves this property. Therefore, $C_{\theta,\gamma}^\ast$ is a bivariate copula. By the standard representation of Kendall's tau in terms of a copula, we obtain $\tau_{00} = 4\int_0^1\int_0^1 C_{\theta,\gamma}^\ast(u,v)\,\mathrm{d}C_{\theta,\gamma}^\ast(u,v) - 1$.
\end{proof}

\subsection{Proof of the result for Spearman's rho in Section 3.2 of the main text} \label{app:proof-spearman}
\begin{proof}
    As with Kendall's tau, we first describe the tie correction induced by the cure fractions. Consider three independent subjects $k$, $l$ and $m$ with observations $\boldsymbol{T}_k=(T_{k1},T_{k2})$, $\boldsymbol{T}_l=(T_{l1},T_{l2})$, and $\boldsymbol{T}_m=(T_{m1},T_{m2})$. When there are no ties, we denote Spearman's rho by $\rho_a$, and its tie-adjusted version by $\rho_b$. Let $\rho_{00}$ denote Spearman's rho within the uncured subpopulation. In the presence of a cure fraction, ties occur through $T=\infty$, and the denominator correction terms are those in \citet{pimentelKendallsTauSpearmans2009}. Hence,
    \begin{equation}
        \rho_b
        = \rho_a \frac{W_0}{\sqrt{(W_0-U_{T_1})(W_0-V_{T_2})}}
        \ \longrightarrow \
        \frac{\rho_a}{\sqrt{(1-p_1^3)(1-p_2^3)}}
        \qquad (n\to\infty),
        \label{eq:rho_b}
    \end{equation}
    where $W_0 = n^3 - n$, $U_{T_1} = (n p_1)^3 - n p_1$, and $V_{T_2} = (n p_2)^3 - n p_2$.
    
    Let $\mathcal{C}_\rho$ and $\mathcal{D}_\rho$ denote the concordant and discordant events:
    \begin{align*}
      \mathcal{C}_\rho
      &\coloneqq \{T_{k1} > T_{l1}, T_{k2} > T_{m2}\} \cup \{T_{k1} < T_{l1}, T_{k2} < T_{m2}\}, \\
      \mathcal{D}_\rho
      &\coloneqq \{T_{k1} > T_{l1}, T_{k2} < T_{m2}\} \cup \{T_{k1} < T_{l1}, T_{k2} > T_{m2}\}.
    \end{align*}
    Since there are three choices for the shared observation among the two comparisons, Spearman's rho without tie correction can be written as
    \[ \rho_a \coloneqq 3 \Prob(\mathcal{C}_\rho) - 3 \Prob(\mathcal{D}_\rho). \]
    
    \begin{table}[htbp]
        \centering
        \caption{Configurations of $(T_{k1},T_{k2})$, $T_{l1}$, and $T_{m2}$}
        \label{tab:spearman_rho}
        \setlength{\tabcolsep}{6pt}
        \renewcommand{\arraystretch}{1.15}
        \begin{tabular}{cccc}
            \toprule
            $(T_{k1},T_{k2})$ & $T_{l1}$ & $T_{m2}$ & $\Prob(\mathcal{C}_\rho)-\Prob(\mathcal{D}_\rho)$ \\
            \midrule
            \multirow{4}{*}{$(T_{k1}<\infty,T_{k2}<\infty)$}
            & $T_{l1}<\infty$ & $T_{m2}<\infty$ & $p_{00}(1-p_1)(1-p_2)\rho_{00}/3$ \\
            & $T_{l1}<\infty$ & $T_{m2}=\infty$ & 0 \\
            & $T_{l1}=\infty$ & $T_{m2}<\infty$ & 0 \\
            & $T_{l1}=\infty$ & $T_{m2}=\infty$ & $p_{00} p_1 p_2$ \\
            \midrule
            \multirow{4}{*}{$(T_{k1}<\infty,T_{k2}=\infty)$}
            & $T_{l1}<\infty$ & $T_{m2}<\infty$ & 0 \\
            & $T_{l1}<\infty$ & $T_{m2}=\infty$ & 0 \\
            & $T_{l1}=\infty$ & $T_{m2}<\infty$ & $- p_{01} p_1 (1 - p_2)$ \\
            & $T_{l1}=\infty$ & $T_{m2}=\infty$ & 0 \\
            \midrule
            \multirow{4}{*}{$(T_{k1}=\infty,T_{k2}<\infty)$}
            & $T_{l1}<\infty$ & $T_{m2}<\infty$ & 0 \\
            & $T_{l1}<\infty$ & $T_{m2}=\infty$ & $- p_{10} (1 - p_1) p_2$ \\
            & $T_{l1}=\infty$ & $T_{m2}<\infty$ & 0 \\
            & $T_{l1}=\infty$ & $T_{m2}=\infty$ & 0 \\
            \midrule
            \multirow{4}{*}{$(T_{k1}=\infty,T_{k2}=\infty)$}
            & $T_{l1}<\infty$ & $T_{m2}<\infty$ & $p_{11} (1 - p_1) (1 - p_2)$ \\
            & $T_{l1}<\infty$ & $T_{m2}=\infty$ & 0 \\
            & $T_{l1}=\infty$ & $T_{m2}<\infty$ & 0 \\
            & $T_{l1}=\infty$ & $T_{m2}=\infty$ & 0 \\
            \bottomrule
        \end{tabular}
    \end{table}
    
    We classify configurations according to whether each of $T_{k1}$, $T_{k2}$, $T_{l1}$, and $T_{m2}$ is finite or infinite. There are $2^4=16$ such configurations, summarized in \Cref{tab:spearman_rho}. For each configuration, the contribution to $\Prob(\mathcal{C}_\rho)-\Prob(\mathcal{D}_\rho)$ is either (i) zero due to a tie or a random ordering with probability $1/2$, or (ii) $\pm 1$ because the ordering is deterministic.

    If $T_{k1}<\infty$, $T_{k2}<\infty$, $T_{l1}<\infty$, $T_{m2}<\infty$, there is no cure-induced tie, so the conditional contribution is $\rho_{00}/3$, with probability $p_{00}(1-p_1)(1-p_2)$. If $T_{k1}<\infty$, $T_{k2}<\infty$, $T_{l1}=\infty$, and $T_{m2}=\infty$, then $T_{k1}<T_{l1}$ and $T_{k2}<T_{m2}$ always hold, so $\mathcal{C}_\rho$ occurs with probability one. This configuration has probability $p_{00}p_1p_2$. The other three cases follow by the same argument.
    
    All remaining configurations contribute zero. For instance, if $T_{k1}<\infty$, $T_{k2}<\infty$, $T_{l1}<\infty$, and $T_{m2}=\infty$,
    then $T_{k2}<T_{m2}$ always holds while $T_{k1}$ and $T_{l1}$ are independent and identically distributed (i.i.d.) and continuous conditional on being finite, so
    \[
        \Prob(\mathcal{C}_\rho) - \Prob(\mathcal{D}_\rho)=\Prob(T_{k1}<T_{l1}) - \Prob(T_{k1}>T_{l1})= \frac12 - \frac12 = 0,
    \]
    and the net contribution is zero. Summing the nonzero contributions in \Cref{tab:spearman_rho} yields
    \begin{align*}
      &\Prob(\mathcal{C}_\rho)-\Prob(\mathcal{D}_\rho) \\
      &= p_{00}p_1p_2 - p_{01}p_1(1-p_2) - p_{10}(1-p_1)p_2
      + p_{11}(1-p_1)(1-p_2) + p_{00}(1-p_1)(1-p_2)\frac{\rho_{00}}{3},
    \end{align*}
    and therefore
    \[
        \rho_a
        = 3\left(p_{11}p_{00}-p_{10}p_{01}\right) + p_{00}(1-p_1)(1-p_2)\rho_{00}.
    \]
    Combining this with \cref{eq:rho_b} proves the formula for $\rho_b$. The representation of $\rho_{00}$ follows by the same argument as in the proof for Kendall's tau in Section S3.1.
\end{proof}

\subsection{Attainable ranges and copula-specific expressions}
\subsubsection{\texorpdfstring{Properties of $\tau_b$}{Properties of taub}}
\label{subsec:property-tau}
The tie-adjusted Kendall's tau $\tau_b$ has a desirable property: it can attain the entire range $(-1,1]$. 

\begin{corollary}[Attainable range of $\tau_b$]
    \label{cor:tau-range}
    The tie-adjusted Kendall's tau $\tau_b$ can take any value in $(-1,1]$.
\end{corollary}

\begin{proof}
It suffices to show that $\tau_b$ can attain any value in $(-1,1]$ under a suitable submodel of the proposed framework. When $\gamma \to 0$, the frailty effect vanishes and the proposed model reduces to the corresponding cure-copula model. In the special case $p_{00}=1$, we have $\tau_b=\tau_{00}$. Under this reduced model, we may choose a comprehensive ordered bivariate copula family, such as the Frank copula family. Since such a family can attain any Kendall's tau value in $(-1,1)$, it follows that $\tau_b$ can take any value in $(-1,1)$. Moreover, for any $\gamma > 0$, if the Fr\'echet--Hoeffding upper bound copula is used as $C_\theta$, the induced copula $C_{\theta,\gamma}^\ast$ remains the upper bound copula, so that $\tau_{00}=1$ and hence $\tau_b=1$. Therefore, $\tau_b$ can take any value in $(-1,1]$.
\end{proof}

\subsubsection*{Example 1: Independence copula}
For the independence copula model, define $u=(1+\gamma r_1 t_1^{a_1})^{-1/\gamma}$, $v=(1+\gamma r_2 t_2^{a_2})^{-1/\gamma}$, which are the marginal survival functions for each margin conditional on $Y_j = 0$ (after integrating out the shared gamma frailty $W$). The last term of the bivariate survival function takes the form of the Clayton copula, whose Kendall’s tau is given by $\tau_{00} = \gamma / (\gamma + 2)$. Thus,
\[ 
    \tau_b
    = \frac{2(p_{11}p_{00} - p_{01}p_{10})}{\sqrt{(1-p_1^2)(1-p_2^2)}}
    + \frac{p_{00}^2}{\sqrt{(1-p_1^2)(1-p_2^2)}}
    \cdot \frac{\gamma}{\gamma + 2}.
\]
Therefore, the attainable range of $\tau_b$ in this model is
\[
    -\frac{2}{3} \le \tau_b \le 1.
\]
The lower bound is attained when $p_{00}=0$ and $p_{10}=p_{01}=1/2$. In this case, there is no uncured pair, and the dependence is entirely determined by the cure-indicator component. The upper bound is attained when $p_{00}=1$ and $\gamma\to\infty$, because then $\tau_{00}=\gamma/(\gamma+2)\to 1$.

\subsubsection*{Example 2: Gumbel copula}
For the Gumbel copula model, let $u$ and $v$ be defined as above. The last term of the bivariate survival function takes the form of the BB1 copula defined in Section 2.1 of the main text, whose Kendall’s tau is given by $\tau_{00} = 1 - 2/\{(\theta + 1)(\gamma + 2)\}$. Thus,
\[
    \tau_b
    = \frac{2(p_{11}p_{00} - p_{01}p_{10})}{\sqrt{(1-p_1^2)(1-p_2^2)}}
    + \frac{p_{00}^2}{\sqrt{(1-p_1^2)(1-p_2^2)}} \cdot
    \left(1 - \frac{2}{(\theta + 1)(\gamma + 2)}\right).
\]
As in Example 1, the attainable range of $\tau_b$ in this model is $-2/3 \leq \tau_b \leq 1$. The upper bound is approached when $p_{00}=1$ and either $\theta\to\infty$ or $\gamma\to\infty$, because then $\tau_{00} \to 1$.

\subsubsection*{Example 3: FGM copula}
For the FGM copula model, let $u$ and $v$ be defined as above. The last term of the bivariate survival function takes the form of the generalized FGM copula described in Section S1.2, whose Kendall’s tau $\tau_{00}$ does not admit a closed-form expression in elementary functions. Nevertheless, its attainable range is known:
\[
  -\frac{2}{9} \leq \tau_{00} \leq 1.
\]
More precisely, $C_{\theta,\gamma}$ reduces to the FGM copula $C_\theta$ as $\gamma \to 0$, so the lower bound $-2/9$ is attained in the limit $\theta=-1$ and $\gamma\to 0$. Similarly, $C_{\theta,\gamma}$ reduces to the Clayton copula $C_\gamma$ when $\theta = 0$, so the upper bound $1$ is attained in the limit $\theta=0$ and $\gamma\to\infty$.

Thus, the attainable range of $\tau_b$ in this model is $-2/3 \leq \tau_b \leq 1$. As in Examples 1 and 2, the lower bound is attained when $p_{00}=0$ and $p_{10}=p_{01}=1/2$. The upper bound is attained in the limiting case where $p_{00}=1$, $\theta=0$, and $\gamma\to\infty$, under which $\tau_{00}\to 1$.

\subsubsection{\texorpdfstring{Properties of $\rho_b$}{Properties of rhob}}
\label{subsec:property-rho}
The tie-adjusted Spearman's rho $\rho_b$ also has a desirable property: it can attain the entire range $(-1,1]$.

\begin{corollary}[Attainable range of $\rho_b$]
    \label{cor:rho-range}
    The tie-adjusted Spearman's rho $\rho_b$ can take any value in $(-1,1]$.
\end{corollary}

\begin{proof}
The proof is analogous to that of \Cref{cor:tau-range}.
\end{proof}

For the independence copula model we have $\rho_{00} = \rho_\mathrm{Clayton}$, and for the Gumbel copula model we have $\rho_{00} = \rho_\mathrm{BB1}$. Unlike Kendall's tau, however, neither $\rho_\mathrm{Clayton}$ nor $\rho_\mathrm{BB1}$ admits a closed-form expression in elementary functions. Consequently, these quantities must be evaluated numerically.

\section{Numerical validation of population rank correlation coefficients}
We numerically validate the rank correlation coefficients for the Gumbel copula model derived in Sections 3.1 and 3.2 of the main text. We generate $1,000$ Monte Carlo replications with sample sizes $n \in \{100, 300, 1,000, 3,000\}$. For replication $k$, we generate one dataset using a prespecified seed $s_k$ and compute both the tie-adjusted sample Kendall's tau and Spearman's rho from that same dataset, denoted by $\widehat{\tau}_b^{(k)}$ and $\widehat{\rho}_b^{(k)}$, respectively. The theoretical values $\tau_b$ and $\rho_b$ are computed using the formulas in Sections 3.1 and 3.2 of the main text; $\rho_{\mathrm{BB1}}$ is evaluated numerically.

For each coefficient $\xi\in\{\tau_b,\rho_b\}$, let $\widehat{\xi}^{(k)}$ denote the estimate from replication $k$, and define the Monte Carlo mean by $\overline{\xi} = \sum_{k=1}^{1,000} \widehat{\xi}^{(k)} / 1,000$. We report $\overline{\xi}$ along with a $95\%$ confidence interval (CI) $\overline{\xi} \pm 1.96\,\mathrm{SE}(\overline{\xi})$, where the standard error (SE) is $\mathrm{SE}(\overline{\xi}) = \mathrm{sd}(\widehat{\xi}^{(1)}, \ldots, \widehat{\xi}^{(1,000)})/\sqrt{1,000}$ and $\mathrm{sd}$ denotes the sample standard deviation.
In addition to the settings reported below, we also considered parameter configurations yielding true values close to 0 and 1, and obtained similar results (not shown). 

We set the model parameters to $\theta=1$, $\gamma=1$, $p_1=0.4$, $p_2=0.2$, and $\OR=2$. Using the formulas in Sections 3.1 and 3.2 of the main text, the theoretical values are $\tau_{\mathrm{BB1}} \approx 0.667$, $\tau_b \approx 0.252$, $\rho_{\mathrm{BB1}} \approx 0.847$, and $\rho_b \approx 0.299$.

\Cref{tab:tau-rho-validation} reports the Monte Carlo mean and $95\%$ CI for each coefficient across varying sample sizes. If the derived formulas are correct and the sample estimators are consistent, we expect the CIs to cover the theoretical values and to become narrower as $n$ increases.

\begin{table}[ht]
    \centering
    \caption{Monte Carlo means and $95\%$ CIs for Kendall's $\tau_b$ and Spearman's $\rho_b$ based on $1{,}000$ replications; the corresponding theoretical values are $0.252$ and $0.299$}
    \label{tab:tau-rho-validation}
    \begin{tabular}{c c c c c}
        \toprule
        & \multicolumn{2}{c}{Kendall's $\tau_b$} & \multicolumn{2}{c}{Spearman's $\rho_b$} \\
        \cmidrule(lr){2-3} \cmidrule(lr){4-5}
        $n$ & mean & $95\%$ CI & mean & $95\%$ CI \\
        \midrule
        100   & 0.251 & [0.245, 0.256] & 0.303 & [0.296, 0.309] \\
        300   & 0.249 & [0.246, 0.252] & 0.300 & [0.296, 0.304] \\
        1,000  & 0.252 & [0.250, 0.254] & 0.300 & [0.298, 0.302] \\
        3,000  & 0.251 & [0.250, 0.252] & 0.299 & [0.297, 0.300] \\
        \bottomrule
    \end{tabular}
\end{table}
As shown in \Cref{tab:tau-rho-validation}, the CIs cover the theoretical values for all sample sizes and become narrower as $n$ increases, confirming the validity of the derived formulas and the consistency of the sample rank correlation coefficients.

\section{Technical details for likelihood-based inference}
\label{supp:loglik}
This section provides the log-likelihood functions for the independence, Gumbel, and FGM copula models with covariates, together with details of standard-error and confidence-interval computation. The likelihood expressions in this section are presented for $\boldsymbol{\xi}_1=\boldsymbol{\xi}_2=\boldsymbol{0}$. For subject $i$ and margin $j$, let $T_{ij}$, $C_{ij}$, and $Y_{ij}$ denote the event time, censoring time, and cure indicator, respectively, and let $\boldsymbol{x}_{ij}$ denote the covariate vector. The observed time is $t_{ij}=\min(T_{ij},C_{ij})$, and $\delta_{ij}=I(T_{ij}\leq C_{ij})$ is the event indicator, where $I(\cdot)$ denotes the indicator function. We write $p_{rs,i}=\Prob(Y_{i1}=r,Y_{i2}=s\mid\boldsymbol{x}_{i1},\boldsymbol{x}_{i2})$ and let $\boldsymbol{\beta}_j$ denote the cure-regression coefficient vector for margin $j$. Finally, $d_j = \sum_{i=1}^n \delta_{ij}$ is the number of observed events for margin $j$ ($j=1,2$), and $d_{12} = \sum_{i=1}^n \delta_{i1} \delta_{i2}$ is the number of subjects with observed events for both margins.

\subsection{Log-likelihood for the independence copula with Weibull margins}
We denote the vector of all parameters by $\boldsymbol{\Psi} = (\gamma, \boldsymbol{\beta}_1, \boldsymbol{\beta}_2, \OR, a_1, r_1, a_2, r_2)$. To simplify the notation, define $B_{ij}$ as follows:
\begin{align*}
    B_{i1}
    &\coloneqq B_{i1}(\boldsymbol{\Psi})
    = p_{01, i} \left(1+\gamma r_1 t_{i1}^{a_1}\right)^{-(1/\gamma+1)}
    + p_{00, i} \left\{1 + \gamma (r_1 t_{i1}^{a_1} + r_2 t_{i2}^{a_2})\right\}^{-(1/\gamma + 1)}, \\
    B_{i2}
    &\coloneqq B_{i2}(\boldsymbol{\Psi})
    = p_{10, i} \left(1+\gamma r_2 t_{i2}^{a_2}\right)^{-(1/\gamma+1)}
    + p_{00, i} \left\{1 + \gamma (r_1 t_{i1}^{a_1} + r_2 t_{i2}^{a_2})\right\}^{-(1/\gamma + 1)}.
\end{align*}
Then, the log-likelihood function is expressed as
\begin{align*}
    \ell(\boldsymbol{\Psi})
    &= d_{12} \log(1 + \gamma) + \sum_{j=1}^{2} d_j \left(\log r_j + \log a_j\right)
    + \sum_{i=1}^n (1-\delta_{i1})(1-\delta_{i2}) \log S(t_{i1}, t_{i2}) \\
    &\quad + \sum_{i=1}^n \delta_{i1} (a_1 - 1) \log t_{i1}
    + \sum_{i=1}^n \delta_{i2} (a_2 - 1) \log t_{i2} \\
    &\quad + \sum_{i=1}^{n} \delta_{i1} (1 - \delta_{i2}) \log B_{i1}
    + \sum_{i=1}^{n} (1 - \delta_{i1}) \delta_{i2} \log B_{i2} \\
    &\quad + \sum_{i=1}^{n} \delta_{i1} \delta_{i2} \left[\log p_{00, i} - \left(\frac{1}{\gamma} + 2\right) \log\left\{1 + \gamma (r_1 t_{i1}^{a_1} + r_2 t_{i2}^{a_2})\right\}\right].
\end{align*}

\subsection{Log-likelihood for the Gumbel copula with Weibull margins}
We derive the log-likelihood for the Gumbel copula model introduced in Section 2.1 of the main text. We adopt Weibull baseline marginal survival functions, so that $\boldsymbol{\phi}_j = (a_j, r_j)$. We present the log-likelihood for the case $\OR \neq 1$, with the full parameter vector $\boldsymbol{\Psi} = (\theta, \gamma, \boldsymbol{\beta}_1, \boldsymbol{\beta}_2, \OR, a_1, r_1, a_2, r_2)$. To simplify notation, let $d_j = \sum_{i=1}^{n} \delta_{ij}$ denote the number of observed events for margin $j$, and $d_{12} = \sum_{i=1}^{n} \delta_{i1} \delta_{i2}$ the number of subjects with observed events on both margins. We further define
\begin{align*}
    A_i
    &\coloneqq A_i(\theta, \gamma, a_1, r_1, a_2, r_2)
    = 1 + \gamma \left\{\left(r_1 t_{i1}^{a_1}\right)^{\theta + 1} + \left(r_2 t_{i2}^{a_2}\right)^{\theta + 1} \right\}^{1/(\theta + 1)}, \\
    B_{i1}
    &\coloneqq B_{i1}(\boldsymbol{\Psi}) \\
    &= \log \left\{p_{01, i} \left(1+\gamma r_1 t_{i1}^{a_1}\right)^{-(1/\gamma+1)}
    + p_{00, i} \gamma^\theta r_1^{\theta} t_{i1}^{a_1 \theta} A_i^{-(1/\gamma+1)}\left(A_i-1\right)^{-\theta}\right\}, \\
    B_{i2}
    &\coloneqq B_{i2}(\boldsymbol{\Psi}) \\
    &= \log \left\{p_{10, i} \left(1+\gamma r_2 t_{i2}^{a_2}\right)^{-(1/\gamma+1)}
    + p_{00, i} \gamma^\theta r_2^{\theta} t_{i2}^{a_2 \theta} A_i^{-(1/\gamma+1)}\left(A_i-1\right)^{-\theta}\right\}.
\end{align*}
Then, the log-likelihood function is expressed as
\begin{align*}
    \ell(\boldsymbol{\Psi})
    &= (2\theta + 1) d_{12} \log \gamma
    + \sum_{j=1}^{2} \left\{d_j \log a_j + (d_j + \theta d_{12}) \log r_j\right\} \\
    &\quad 
    + \sum_{i=1}^n \delta_{i1}\left(a_1 - 1 + \delta_{i2}a_1 \theta\right) \log t_{i1}
    + \sum_{i=1}^n \delta_{i2}\left(a_2 - 1 + \delta_{i1}a_2 \theta\right) \log t_{i2} \\
    &\quad + \sum_{i=1}^n \delta_{i1} (1-\delta_{i2}) B_{i1}
    + \sum_{i=1}^n (1-\delta_{i1}) \delta_{i2} B_{i2}
    + \sum_{i=1}^n (1-\delta_{i1})(1-\delta_{i2}) \log S(t_{i1}, t_{i2}) \\
    &\quad + \sum_{i=1}^n \delta_{i1}\delta_{i2}
    \log\left\{ \left(1 + \frac{1}{\gamma} + \theta\right)A_i - \left(1 + \frac{1}{\gamma}\right) \right\} \\
    &\quad + \sum_{i=1}^n \delta_{i1}\delta_{i2} \left\{
    \log p_{00, i} - \left(2 + \frac{1}{\gamma}\right)\log A_i
    - \left(2\theta + 1\right)\log\left(A_i - 1\right)\right\}.
\end{align*}
When $\OR=1$, the expression remains unchanged except that $p_{rs,i}$ is replaced by the independence probabilities, e.g., $p_{11,i}=p_1(\boldsymbol{x}_{i1})p_2(\boldsymbol{x}_{i2})$ (and similarly for the other $p_{rs,i}$). In the limiting case $\OR=\infty$, the model simplifies further through the induced constraints on $p_{rs,i}$ and the use of a common regression parameter $\boldsymbol{\beta}$; we omit the explicit log-likelihood for brevity.

\subsection{Log-likelihood for the FGM copula with Weibull margins}
We denote the vector of all parameters by $\boldsymbol{\Psi} = (\theta, \gamma, \boldsymbol{\beta}_1, \boldsymbol{\beta}_2, \OR, a_1, r_1, a_2, r_2)$. To simplify the notation, let $q_{i1} = 1 + \gamma r_1 t_{i1}^{a_1}$ and $q_{i2} = 1 + \gamma r_2 t_{i2}^{a_2}$, and define
\begin{align*}
    &D_{i1}
    \coloneqq D_{i1}(\boldsymbol{\Psi}) \\
    &= p_{01, i} q_{i1}^{-(1/\gamma+1)}
    + p_{00, i} \left\{(1+\theta)\left(q_{i1} + q_{i2} - 1\right)^{-(1/\gamma + 1)}
    - 2 \theta \left(2 q_{i1} + q_{i2} - 2\right)^{-(1/\gamma + 1)}
    \right. \\
    & \quad \left.
    - \theta \left(q_{i1} + 2 q_{i2} - 2\right)^{-(1/\gamma + 1)}
    + 2 \theta \left(2q_{i1} + 2q_{i2} - 3\right)^{-(1/\gamma + 1)}\right\}, \\
    &D_{i2}
    \coloneqq D_{i2}(\boldsymbol{\Psi}) \\
    &= p_{10, i} q_{i2}^{-(1/\gamma+1)}
    + p_{00, i} \left\{(1+\theta)\left(q_{i1} + q_{i2} - 1\right)^{-(1/\gamma + 1)}
    - \theta \left(2 q_{i1} + q_{i2} - 2\right)^{-(1/\gamma + 1)}
    \right. \\
    &\quad \left.
    - 2 \theta \left(q_{i1} + 2 q_{i2} - 2\right)^{-(1/\gamma + 1)}
    + 2 \theta \left(2q_{i1} + 2q_{i2} - 3\right)^{-(1/\gamma + 1)}\right\}, \\
    &E_i
    \coloneqq E_i (\boldsymbol{\Psi}) \\
    & = \left\{(1 + \theta) \left(q_{i1} + q_{i2} - 1\right)^{-(1/\gamma + 2)}
    - 2 \theta \left(2 q_{i1} + q_{i2} - 2\right)^{-(1/\gamma + 2)} \right. \\
    &\left. \quad
    - 2 \theta \left(q_{i1} + 2 q_{i2} - 2\right)^{-(1/\gamma + 2)}
    + 4 \theta \left(2q_{i1} + 2q_{i2} - 3\right)^{-(1/\gamma + 2)}\right\}.
\end{align*}
Then, the log-likelihood function is expressed as
\begin{align*}
    \ell(\boldsymbol{\Psi})
    &= d_{12} \log(1 + \gamma)
    + \sum_{j=1}^{2} d_j \left(\log r_j + \log a_j\right)
    + \sum_{i=1}^n (1-\delta_{i1})(1-\delta_{i2}) \log S(t_{i1}, t_{i2}) \\
    &\, + \sum_{i=1}^{n} \delta_{i1} \delta_{i2} \left(\log p_{00, i} + \log E_i\right) + \sum_{i=1}^n \delta_{i1} \left(a_1 - 1\right) \log t_{i1}
    + \sum_{i=1}^n \delta_{i2}\,\left(a_2 - 1\right) \log t_{i2} \\
    &\, + \sum_{i=1}^{n} \delta_{i1} (1 - \delta_{i2}) \log D_{i1}
    + \sum_{i=1}^{n} (1 - \delta_{i1}) \delta_{i2} \log D_{i2}.
\end{align*}

\subsection{Parameter dimensions, standard errors, and confidence intervals}
The parameter vector depends on the cure-dependence specification. For $0<\OR<1$ and $\OR>1$, we use the full parameter vector $\boldsymbol{\Psi}=(\theta,\gamma,\boldsymbol{\beta}_1,\boldsymbol{\beta}_2,\OR,\boldsymbol{\phi}_1,\boldsymbol{\phi}_2)$, whose dimension is $3+\dim(\boldsymbol{\beta}_1)+\dim(\boldsymbol{\beta}_2)+\dim(\boldsymbol{\phi}_1)+\dim(\boldsymbol{\phi}_2)$. When $\OR=1$, the joint cure-status probabilities reduce to the independence case and the model no longer depends on $\OR$; accordingly, the parameter vector becomes $(\theta,\gamma,\boldsymbol{\beta}_1,\boldsymbol{\beta}_2,\boldsymbol{\phi}_1,\boldsymbol{\phi}_2)$. In the limiting case $\OR=\infty$, the joint cure status is fully determined and a common regression parameter $\boldsymbol{\beta}$ is used for the marginal cure fractions, as described in Section 2.3 of the main text. The effective parameter vector is then $(\theta,\gamma,\boldsymbol{\beta},\boldsymbol{\phi}_1,\boldsymbol{\phi}_2)$, with dimension $2+\dim(\boldsymbol{\beta})+\dim(\boldsymbol{\phi}_1)+\dim(\boldsymbol{\phi}_2)$.

By the asymptotic normality of the maximum likelihood estimator \citep{TheoryPointEstimation1998}, $\widehat{\boldsymbol{\Psi}}$ is approximately multivariate normal with mean equal to the true parameter vector $\boldsymbol{\Psi}_0$ and covariance matrix $\mathcal{I}_n(\boldsymbol{\Psi}_0)^{-1}$, where $\mathcal{I}_n(\boldsymbol{\Psi}_0)$ denotes the Fisher information matrix. In computation, we approximate $\mathcal{I}_n(\boldsymbol{\Psi}_0)$ by the observed Fisher information matrix $\mathcal{J}(\widehat{\boldsymbol{\Psi}})=-\partial^2 \ell(\widehat{\boldsymbol{\Psi}}) / \partial \boldsymbol{\Psi} \partial \boldsymbol{\Psi}^\top$, so that $\widehat{\mathrm{Var}}(\widehat{\boldsymbol{\Psi}}) \approx \mathcal{J}(\widehat{\boldsymbol{\Psi}})^{-1}$.

When $\OR$ is fixed at $\OR=1$ or $\OR=\infty$, it is not estimated and no confidence interval (CI) is reported. When $\OR$ is estimated, we construct a Wald-type CI on an unconstrained scale using the region-specific transformation
\[
g(\OR)=
\begin{cases}
\log\left(\dfrac{\OR}{1-\OR}\right), & 0<\OR<1, \\
\log\left(\OR-1\right), & \OR>1,
\end{cases}
\]
and back-transform the endpoints using $g^{-1}$. Since $\theta$, $\gamma$, $a_j$, and $r_j$ are strictly positive, we similarly construct Wald-type CIs on the log scale and back-transform. Specifically, for a scalar parameter $\psi$ with support $(0,\infty)$ and significance level $\alpha$, we use the log transformation and obtain an approximate $100(1-\alpha)\%$ CI as
\[
    \left[\exp\left\{\log(\widehat{\psi})-z_{\alpha/2}\mathrm{SE}\left(\log(\widehat{\psi})\right)\right\},\,\exp\left\{\log(\widehat{\psi})+z_{\alpha/2}\mathrm{SE}\left(\log(\widehat{\psi})\right)\right\}\right],
\]
where $z_{\alpha/2}$ is the upper $\alpha/2$ quantile of the standard normal distribution and $\mathrm{SE}\left(\log(\widehat{\psi})\right)$ is obtained using the delta method.

\section{Additional simulation details and results}
\label{supp:sim}

\subsection{Data-generation algorithm for the simulation with covariates}
\label{app:sim-algorithm}

The algorithm used in Section 6.1 of the main text to generate pairs of observed survival times and censoring indicators $\{t_{i1}, t_{i2}, \delta_{i1}, \delta_{i2}\}_{i=1}^n$ proceeds as follows. We use $x_{ij}^{\mathrm{cure}}$ and $x_{ij}^{\mathrm{surv}}$ to denote the covariates in the cure and susceptible survival-time components, respectively.
\begin{enumerate}
  \item For $i=1,\ldots,n$ and $j=1,2$, draw $x_{ij}^{\mathrm{cure}}\overset{\text{i.i.d.}}{\sim}\mathrm{Bernoulli}(0.5)$ and $x_{ij}^{\mathrm{surv}}\overset{\text{i.i.d.}}{\sim}\mathrm{Uniform}(0,1)$. The covariates are generated independently across the two model components and margins.
  \item Compute the subject-specific cure fractions using
  \[
    p_j\left(x_{ij}^{\mathrm{cure}}\right)
    = \frac{\exp\left(\beta_{0j}^{\mathrm{cure}}+\beta_j^{\mathrm{cure}}x_{ij}^{\mathrm{cure}}\right)}
    {1+\exp\left(\beta_{0j}^{\mathrm{cure}}+\beta_j^{\mathrm{cure}}x_{ij}^{\mathrm{cure}}\right)},
    \qquad j=1,2.
  \]
  \item Draw cure indicators $(Y_{i1},Y_{i2})$:
  \subitem i. If $\OR=1$, draw $Y_{i1}\sim\mathrm{Bernoulli}\{p_1(x_{i1}^{\mathrm{cure}})\}$ and $Y_{i2}\sim\mathrm{Bernoulli}\{p_2(x_{i2}^{\mathrm{cure}})\}$ independently.
  \subitem ii. If $0<\OR<\infty$ and $\OR\neq 1$, draw $(Y_{i1},Y_{i2})$ from a multinomial distribution over $\{(1,1),(1,0),(0,1),(0,0)\}$ with cell probabilities $(p_{11,i},p_{10,i},p_{01,i},p_{00,i})$. These probabilities are defined in Section 2.3 of the main text as functions of $p_1(x_{i1}^{\mathrm{cure}})$, $p_2(x_{i2}^{\mathrm{cure}})$, and $\OR$.
  \subitem iii. If $\OR=\infty$, then $Y_{i1}=Y_{i2}$ almost surely, so we require $x_{i1}^{\mathrm{cure}}=x_{i2}^{\mathrm{cure}}\eqqcolon x_i^{\mathrm{cure}}$ and $p_1(x_i^{\mathrm{cure}})=p_2(x_i^{\mathrm{cure}})$. Draw $Y_{i1}\sim\mathrm{Bernoulli}\{p_1(x_i^{\mathrm{cure}})\}$ and set $Y_{i2}=Y_{i1}$.
  \item Draw $W_i\sim\mathrm{Gamma}(1/\gamma,\gamma)$ and $Z_{ij}=(1-Y_{ij})W_i$.
  \item Draw $(U_{i1},U_{i2})$ from the Gumbel copula with parameter $\theta$.
  \item Under the proportional hazards specification
  \[
    S_{0j}(t\mid x_{ij}^{\mathrm{surv}})
    =\exp\left\{-r_jt^{a_j} \exp\left(\beta_j^{\mathrm{surv}}x_{ij}^{\mathrm{surv}}\right) \right\},
  \]
  define the event times by
  \[
    T_{ij}=
    \begin{cases}
      \left\{\dfrac{-\log U_{ij}}{r_j Z_{ij}\exp\left(\beta_j^{\mathrm{surv}}x_{ij}^{\mathrm{surv}}\right)}\right\}^{1/a_j} & \text{if } Z_{ij}>0,\\
      \infty & \text{if } Z_{ij}=0.
    \end{cases}
  \]
  \item Draw $C_i\sim\mathrm{Uniform}(0,6)$ and set $t_{ij}=\min\{T_{ij},C_i\}$, $\delta_{ij}=I(T_{ij}\le C_i)$.
\end{enumerate}

\subsection{Simulation study without covariates}

In this section, we present a simulation study to evaluate the finite-sample performance of the proposed model without covariates. We assess the accuracy of parameter estimation under different dependence structures.

\subsubsection{Simulation design}
We generate synthetic datasets from the proposed model. Since covariates are not incorporated, the parameter vector is $\boldsymbol{\Psi}^\ast = (\theta, \gamma, p_1, p_2, \OR, a_1, r_1, a_2, r_2)$; see Section 4 of the main text. The algorithm to generate pairs of survival times and censoring indicators $\{t_{i1}, t_{i2}, \delta_{i1}, \delta_{i2}\}_{i=1}^n$ proceeds as follows:
\begin{enumerate}
  \item Draw cure indicators $(Y_{i1}, Y_{i2})$:
  \subitem i. If $\OR = 1$, draw $Y_{i1} \sim \mathrm{Bernoulli}(p_1)$ and $Y_{i2} \sim \mathrm{Bernoulli}(p_2)$ independently.
  \subitem ii. If $0<\OR<\infty$ and $\OR\neq 1$, draw $(Y_{i1},Y_{i2})$ from a multinomial distribution over $\{(1,1),(1,0),(0,1),(0,0)\}$ with cell probabilities $(p_{11},p_{10},p_{01},p_{00})$.
  \subitem iii. If $\OR=\infty$, set $p=p_1=p_2$, draw $Y_{i1}\sim\mathrm{Bernoulli}(p)$, and set $Y_{i2} = Y_{i1}$.
  \item Draw $W_i \sim \mathrm{Gamma}(1/\gamma, \gamma)$ and $Z_{ij} = (1 - Y_{ij}) W_i$.
  \item Draw $(U_{i1}, U_{i2})$ from the Gumbel copula with parameter $\theta$.
  \item Define event times by inverting the conditional survival function:
  \[
  T_{ij} =
  \begin{cases}
      \left(- \dfrac{\log U_{ij}}{r_j Z_{ij}}\right)^{1/a_j} & \text{if } Z_{ij} > 0, \\
      \infty & \text{if } Z_{ij} = 0.
  \end{cases}
  \]
  \item Draw $C_i \sim \mathrm{Uniform}(0, 6)$ and set $t_{ij} = \min(T_{ij}, C_i)$, $\delta_{ij} = I(T_{ij} \leq C_i)$.
\end{enumerate}

We consider two parameter settings with sample sizes $n \in \{200, 400\}$. Setting A corresponds to positive dependence between cure indicators ($\OR = 2$), with
\begin{equation*}
  (\theta, \gamma, p_1, p_2, \OR, a_1, r_1, a_2, r_2)
  = (2, 0.5, 0.6, 0.4, 2, 1, 1.5, 1, 2),
\end{equation*}
whereas Setting B represents negative cure dependence ($\OR = 0.5$), with
\begin{equation*}
  (\theta, \gamma, p_1, p_2, \OR, a_1, r_1, a_2, r_2)
  = (0.5, 0.5, 0.6, 0.4, 0.5, 1, 1.5, 1, 2).
\end{equation*}
In Settings A and B, the values of $\theta$ and $\OR$ differ. Under this censoring mechanism and the parameter settings above, the average censoring rates are approximately $67\%$ for margin 1 and $49\%$ for margin 2 in both settings. Maximum likelihood estimation is performed using the \texttt{optim} function in \textsf{R}. For each scenario, we conduct $1,000$ Monte Carlo replications.

\subsubsection{Results}
For each parameter, we summarize the mean estimate, bias, mean squared error (MSE), standard deviation (SD), average standard error (SE (mean)), and coverage probability (CP) of the $95\%$ CIs.
\begin{table}[htbp]
  \centering
  \caption{Estimation performance under Settings A and B for sample sizes $n=200$ and $n=400$}
  \label{tab:settings_ab}
  \small
  \begin{tabular}{c c c c c c c c c}
    \multicolumn{9}{c}{Setting A: $\theta = 2$, $\OR = 2$; $n=200,400$} \\
    \midrule
    {} & Parameter & {True} & {Mean} & {Bias} & {MSE} & {SD} & {SE (mean)} & {CP} \\
    \midrule
    \multirow[t]{9}{*}{$n=200$}
    & $\theta$   & 2.0 & 2.047 & 0.047 & 0.351 & 0.591 & 0.581 & 0.953 \\
    & $\gamma$   & 0.5 & 0.679 & 0.179 & 0.472 & 0.664 & 0.588 & 0.923 \\
    & $p_1$      & 0.6 & 0.592 & $-0.008$ & 0.003 & 0.050 & 0.049 & 0.947 \\
    & $p_2$      & 0.4 & 0.390 & $-0.010$ & 0.003 & 0.054 & 0.056 & 0.947 \\
    & $\OR$      & 2.0 & 2.060 & 0.060 & 0.758 & 0.869 & 0.867 & 0.966 \\
    & $a_1$      & 1.0 & 1.048 & 0.048 & 0.033 & 0.176 & 0.161 & 0.911 \\
    & $r_1$      & 1.5 & 1.718 & 0.218 & 0.717 & 0.819 & 0.601 & 0.899 \\
    & $a_2$      & 1.0 & 1.049 & 0.049 & 0.032 & 0.171 & 0.156 & 0.911 \\
    & $r_2$      & 2.0 & 2.382 & 0.382 & 2.162 & 1.420 & 0.952 & 0.892 \\
    \midrule
    \multirow[t]{9}{*}{$n=400$}
    & $\theta$   & 2.0 & 2.023 & 0.023 & 0.178 & 0.421 & 0.414 & 0.943 \\
    & $\gamma$   & 0.5 & 0.605 & 0.105 & 0.215 & 0.452 & 0.420 & 0.952 \\
    & $p_1$      & 0.6 & 0.596 & $-0.004$ & 0.001 & 0.033 & 0.035 & 0.963 \\
    & $p_2$      & 0.4 & 0.392 & $-0.008$ & 0.002 & 0.040 & 0.039 & 0.950 \\
    & $\OR$      & 2.0 & 1.991 & $-0.009$ & 0.362 & 0.602 & 0.596 & 0.960 \\
    & $a_1$      & 1.0 & 1.025 & 0.025 & 0.015 & 0.119 & 0.113 & 0.933 \\
    & $r_1$      & 1.5 & 1.596 & 0.096 & 0.167 & 0.397 & 0.356 & 0.918 \\
    & $a_2$      & 1.0 & 1.026 & 0.026 & 0.014 & 0.116 & 0.110 & 0.920 \\
    & $r_2$      & 2.0 & 2.150 & 0.150 & 0.380 & 0.598 & 0.534 & 0.907 \\
    \midrule
    \addlinespace[0.8em]
    \multicolumn{9}{c}{Setting B: $\theta = 0.5$, $\OR = 0.5$; $n=200,400$} \\
    \midrule
    {} & Parameter & {True} & {Mean} & {Bias} & {MSE} & {SD} & {SE (mean)} & {CP} \\
    \midrule
    \multirow[t]{9}{*}{$n=200$}
    & $\theta$   & 0.5 & 0.532 & 0.032 & 0.101 & 0.316 & 0.290 & 0.939 \\
    & $\gamma$   & 0.5 & 0.681 & 0.181 & 0.491 & 0.677 & 0.565 & 0.919 \\
    & $p_1$      & 0.6 & 0.592 & $-0.008$ & 0.003 & 0.051 & 0.052 & 0.944 \\
    & $p_2$      & 0.4 & 0.384 & $-0.016$ & 0.003 & 0.057 & 0.057 & 0.947 \\
    & $\OR$      & 0.5 & 0.485 & $-0.015$ & 0.052 & 0.227 & 0.229 & 0.988 \\
    & $a_1$      & 1.0 & 1.046 & 0.046 & 0.033 & 0.177 & 0.163 & 0.937 \\
    & $r_1$      & 1.5 & 1.712 & 0.212 & 0.671 & 0.791 & 0.591 & 0.931 \\
    & $a_2$      & 1.0 & 1.040 & 0.040 & 0.028 & 0.163 & 0.150 & 0.918 \\
    & $r_2$      & 2.0 & 2.288 & 0.288 & 1.092 & 1.005 & 0.817 & 0.912 \\
    \midrule
    \multirow[t]{9}{*}{$n=400$}
    & $\theta$   & 0.5 & 0.533 & 0.033 & 0.045 & 0.209 & 0.209 & 0.943 \\
    & $\gamma$   & 0.5 & 0.555 & 0.055 & 0.171 & 0.410 & 0.373 & 0.963 \\
    & $p_1$      & 0.6 & 0.596 & $-0.004$ & 0.001 & 0.033 & 0.035 & 0.972 \\
    & $p_2$      & 0.4 & 0.395 & $-0.005$ & 0.002 & 0.039 & 0.038 & 0.965 \\
    & $\OR$      & 0.5 & 0.494 & $-0.006$ & 0.025 & 0.159 & 0.158 & 0.989 \\
    & $a_1$      & 1.0 & 1.015 & 0.015 & 0.013 & 0.114 & 0.112 & 0.951 \\
    & $r_1$      & 1.5 & 1.566 & 0.066 & 0.162 & 0.397 & 0.342 & 0.931 \\
    & $a_2$      & 1.0 & 1.012 & 0.012 & 0.012 & 0.109 & 0.104 & 0.933 \\
    & $r_2$      & 2.0 & 2.080 & 0.080 & 0.264 & 0.508 & 0.474 & 0.936 \\
    \bottomrule
  \end{tabular}
\end{table}

\Cref{tab:settings_ab} summarizes the estimation results for Settings A and B. The overall patterns are similar across the two settings. As the sample size increases from $n = 200$ to $n = 400$, bias and standard deviation generally decrease, and the coverage probabilities approach the nominal $95\%$ level.

The cure fractions $(p_1, p_2)$ are estimated accurately even at $n = 200$, with small bias and coverage probabilities close to $95\%$. The odds ratio $\OR$ and the copula parameter $\theta$ are estimated with reasonable accuracy, while the frailty parameter $\gamma$ and the Weibull scale parameters $(r_1,r_2)$ show relatively larger bias and variability, especially in the smaller sample.

Overall, the proposed estimation procedure performs reasonably well for moderate sample sizes, with clear improvement when the sample size increases to $n=400$.

\section{EM algorithm versus direct optimization}
\label{supp:em_vs_direct}
In frailty models, the frailty is typically treated as latent (unobserved) data, and parameter estimation is often carried out via the EM algorithm \citep{duchateauSharedFrailtyModel2002,kleinSurvivalAnalysisTechniques2005,duchateauFrailtyModel2008}. In this section, we develop an EM algorithm for the Gumbel copula model with Weibull marginals introduced in Section 2.1 of the main text, and compare its performance with the direct optimization approach in Section 4.1 of the main text. In contrast to the standard frailty setting, we treat the cure indicators $Y_{ij}$ ($i=1,\ldots,n;\, j=1,2$) as the latent (unobserved) data in the EM algorithm.

Throughout this section, we consider the model without covariates and restrict attention to the case $\OR=1$ for simplicity. The parameter vector is accordingly $\boldsymbol{\Psi}^\ast = (\theta, \gamma, p_1, p_2, a_1, r_1, a_2, r_2)$. The extension to the case with covariates or $\OR\neq 1$ is straightforward, though the notation becomes more involved.

\subsection{Complete-data log-likelihood}
For subject $i$, the observed data consist of the survival times and event indicators for the two margins:
\begin{align*}
    T_{ij} &: \text{observed survival time for margin } j, \\
    &\text{i.e., the minimum of the true event time and the censoring time}, \\
    \delta_{ij} &=
    \begin{cases}
        1, & \text{if the event is observed},\\
        0, & \text{if the observation is censored},
    \end{cases}
    \quad j=1,2.
\end{align*}

Let $(T_{i1},T_{i2})$ and $(Y_{i1},Y_{i2})$ denote random pairs, with realizations $(t_{i1},t_{i2})$ and $(y_{i1},y_{i2})$, respectively. Define
\begin{align*}
    &S(t_{i1},t_{i2},y_{i1},y_{i2}) \\
    &\coloneqq \Prob\left(T_{i1}>t_{i1},\,T_{i2}>t_{i2},\,Y_{i1}=y_{i1},\,Y_{i2}=y_{i2}\right) \\
    &=
    (p_1p_2)^{y_{i1}y_{i2}}
    \left\{(1-p_1)p_2(1+\gamma r_1 t_{i1}^{a_1})^{-1/\gamma}\right\}^{(1-y_{i1})y_{i2}} \\
    &\quad\times
    \left\{p_1(1-p_2)(1+\gamma r_2 t_{i2}^{a_2})^{-1/\gamma}\right\}^{y_{i1}(1-y_{i2})}
    \times \left\{(1-p_1)(1-p_2)A_i^{-1/\gamma}\right\}^{(1-y_{i1})(1-y_{i2})},
\end{align*}
where
\[
    A_i \coloneqq
    1+\gamma\left\{\left(r_1 t_{i1}^{a_1}\right)^{\theta+1}
    +\left(r_2 t_{i2}^{a_2}\right)^{\theta+1}\right\}^{1/(\theta+1)}.
\]

We adopt the same notation as in Section 4.1 of the main text, except that the odds-ratio parameter is absent here because we assume $\OR=1$. Accordingly, the full parameter vector is $\boldsymbol{\Psi} = (\theta, \gamma, p_1, p_2, a_1, r_1, a_2, r_2)$. Let $\ell_c(\boldsymbol{\Psi})$ denote the complete-data log-likelihood based on $\{(T_{i1},T_{i2},\delta_{i1},\delta_{i2},Y_{i1},Y_{i2})\}_{i=1}^n$.

In the E-step of the EM algorithm, given the current parameter estimate $\boldsymbol{\Psi}^{(k)}$ at iteration $k$, we evaluate the Q-function
\begin{equation}
    Q(\boldsymbol{\Psi}\mid \boldsymbol{\Psi}^{(k)})
    \coloneqq
    \E_{\boldsymbol{\Psi}^{(k)}}\left[
        \ell_c(\boldsymbol{\Psi})
        \,\middle|\,
        \{(T_{i1},T_{i2},\delta_{i1},\delta_{i2})\}_{i=1}^n
    \right].
    \label{eq:q-function}
\end{equation}
Here, the expectation is taken with respect to the conditional distribution of the latent variables given the observed data under the current parameter value $\boldsymbol{\Psi}^{(k)}$. For notational simplicity, we suppress the subscript $\boldsymbol{\Psi}^{(k)}$ in what follows, and all conditional expectations are understood to be taken under the current parameter value $\boldsymbol{\Psi}^{(k)}$. Expanding \cref{eq:q-function} according to the four possible observation patterns yields
\begin{align}
    Q(\boldsymbol{\Psi}\mid \boldsymbol{\Psi}^{(k)})
    &=
    \sum_{i=1}^n (1-\delta_{i1})(1-\delta_{i2})
    \E\left[
        \log S(T_{i1},T_{i2},Y_{i1},Y_{i2})
        \,\middle|\,
        T_{i1},T_{i2},\delta_{i1},\delta_{i2}
    \right] \notag\\
    &\quad+
    \sum_{i=1}^n \delta_{i1}(1-\delta_{i2})
    \E\left[
        \log\left\{
            -\frac{\partial}{\partial t_{i1}}
            S(T_{i1},T_{i2},Y_{i1},Y_{i2})
        \right\}
        \,\middle|\,
        T_{i1},T_{i2},\delta_{i1},\delta_{i2}
    \right] \notag\\
    &\quad+
    \sum_{i=1}^n (1-\delta_{i1})\delta_{i2}
    \E\left[
        \log\left\{
            -\frac{\partial}{\partial t_{i2}}
            S(T_{i1},T_{i2},Y_{i1},Y_{i2})
        \right\}
        \,\middle|\,
        T_{i1},T_{i2},\delta_{i1},\delta_{i2}
    \right] \notag\\
    &\quad+
    \sum_{i=1}^n \delta_{i1}\delta_{i2}
    \E\left[
        \log\left\{
            \frac{\partial^2}{\partial t_{i1}\partial t_{i2}}
            S(T_{i1},T_{i2},Y_{i1},Y_{i2})
        \right\}
        \,\middle|\,
        T_{i1},T_{i2},\delta_{i1},\delta_{i2}
    \right].
    \label{eq:complete-exp}
\end{align}

To evaluate \cref{eq:complete-exp}, we compute the relevant conditional expectations for each of the four observation patterns. For $(\delta_{i1},\delta_{i2})\in\{0,1\}^2$, define
\begin{align*}
  \hat{y}_{ij}^{(\delta_{i1},\delta_{i2})}
  &\coloneqq
  \E\left[Y_{ij}\mid \text{observed data under }(\delta_{i1},\delta_{i2})\right],
  \quad j=1,2, \\
  \widehat{y_{i1}y_{i2}}^{(\delta_{i1},\delta_{i2})}
  &\coloneqq
  \E\left[Y_{i1}Y_{i2}\mid \text{observed data under }(\delta_{i1},\delta_{i2})\right].
\end{align*}

When both margins are censored, we have
\begin{align*}
    \hat{y}_{i1}^{(0,0)}
    &= \E[Y_{i1}\mid T_{i1}>t_{i1},\,T_{i2}>t_{i2}]
    = \sum_{y_{i1}=0}^1 y_{i1}\,
    \Prob(Y_{i1}=y_{i1}\mid T_{i1}>t_{i1},\,T_{i2}>t_{i2}) \\
    &= \frac{p_1p_2+p_1(1-p_2)(1+\gamma r_2 t_{i2}^{a_2})^{-1/\gamma}}{S(t_{i1}, t_{i2})},
\end{align*}
where $S(t_{i1}, t_{i2}) = p_1p_2+(1-p_1)p_2(1+\gamma r_1 t_{i1}^{a_1})^{-1/\gamma}+p_1(1-p_2)(1+\gamma r_2 t_{i2}^{a_2})^{-1/\gamma}+(1-p_1)(1-p_2)A_i^{-1/\gamma}$ and similarly,
\[
    \hat{y}_{i2}^{(0,0)}
    = \frac{p_1p_2+(1-p_1)p_2(1+\gamma r_1 t_{i1}^{a_1})^{-1/\gamma}}{S(t_{i1}, t_{i2})},
    \qquad
    \widehat{y_{i1}y_{i2}}^{(0,0)}
    = \frac{p_1p_2}{S(t_{i1}, t_{i2})}.
\]
Next, consider the cases in which exactly one margin is uncensored. If an event is observed in one margin, then the corresponding cure indicator must be zero. Therefore, in the case $(\delta_{i1},\delta_{i2})=(0,1)$, we have $\hat{y}_{i2}^{(0,1)}=\widehat{y_{i1}y_{i2}}^{(0,1)}=0$, and, by symmetry, in the case $(\delta_{i1},\delta_{i2})=(1,0)$, we have $\hat{y}_{i1}^{(1,0)}=\widehat{y_{i1}y_{i2}}^{(1,0)}=0$. The remaining conditional expectations, $\hat{y}_{i1}^{(0,1)}$ and $\hat{y}_{i2}^{(1,0)}$, are obtained analogously by Bayes' theorem.

Finally, when events are observed in both margins, we have $Z_{i1}\neq 0$ and $Z_{i2}\neq 0$, which implies $Y_{i1}=Y_{i2}=0$. Therefore, $\hat{y}_{i1}^{(1,1)}=\hat{y}_{i2}^{(1,1)}=\widehat{y_{i1}y_{i2}}^{(1,1)}=0$.

\begin{table}[htbp]
    \centering
    \caption{Conditional expectations required in the E-step under each observation pattern, where $(0,0)$, $(0,1)$, $(1,0)$, and $(1,1)$ denote the realizations of $(\delta_{i1},\delta_{i2})$}
    \begin{tabular}{c c c c c}
        \toprule
        & $(0,0)$ & $(0,1)$ & $(1,0)$ & $(1,1)$ \\
        \midrule
        \rule[-3pt]{0pt}{15pt}$\E[Y_{i1}\mid \text{data}]$
        & $\hat{y}_{i1}^{(0,0)}$ & $\hat{y}_{i1}^{(0,1)}$ & $0$ & $0$ \\
        \rule[-3pt]{0pt}{15pt}$\E[Y_{i2}\mid \text{data}]$
        & $\hat{y}_{i2}^{(0,0)}$ & $0$ & $\hat{y}_{i2}^{(1,0)}$ & $0$ \\
        \rule[-3pt]{0pt}{15pt}$\E[Y_{i1}Y_{i2}\mid \text{data}]$
        & $\widehat{y_{i1}y_{i2}}^{(0,0)}$ & $0$ & $0$ & $0$ \\
        \bottomrule
    \end{tabular}
    \label{tab:estep-summary}
\end{table}
\Cref{tab:estep-summary} summarizes the conditional expectations required in the E-step. In particular, the only nonzero conditional expectations are
$\hat{y}_{i1}^{(0,0)}$, $\hat{y}_{i2}^{(0,0)}$, $\widehat{y_{i1}y_{i2}}^{(0,0)}$, $\hat{y}_{i1}^{(0,1)}$, and $\hat{y}_{i2}^{(1,0)}$; all others are zero. Hence,
\begin{align*}
    &\E\left[\log S(T_{i1},T_{i2},Y_{i1},Y_{i2})
    \mid T_{i1}>t_{i1},\,T_{i2}>t_{i2}\right] \\
    &= \hat{y}_{i1}^{(0,0)}\log p_1
    + \hat{y}_{i2}^{(0,0)}\log p_2
    + \left(1-\hat{y}_{i1}^{(0,0)}\right)\log(1-p_1)
    + \left(1-\hat{y}_{i2}^{(0,0)}\right)\log(1-p_2) \\
    &\quad - \frac{\hat{y}_{i2}^{(0,0)}-\widehat{y_{i1}y_{i2}}^{(0,0)}}{\gamma}
    \log(1+\gamma r_1 t_{i1}^{a_1})
    - \frac{\hat{y}_{i1}^{(0,0)}-\widehat{y_{i1}y_{i2}}^{(0,0)}}{\gamma}
    \log(1+\gamma r_2 t_{i2}^{a_2}) \\
    &\quad - \frac{1-\hat{y}_{i1}^{(0,0)}-\hat{y}_{i2}^{(0,0)}+\widehat{y_{i1}y_{i2}}^{(0,0)}}{\gamma} \log A_i.
\end{align*}
Similarly, the remaining terms in \cref{eq:complete-exp} can be evaluated using
$\hat{y}_{i1}^{(0,1)}$, $\hat{y}_{i2}^{(1,0)}$, and the fact that all other relevant conditional expectations are zero.

\subsection{EM algorithm}

We outline the EM algorithm for estimating the parameter vector $\boldsymbol{\Psi}$.

\begin{itemize}
    \item[Step 0.] Choose an initial value
    \[
        \boldsymbol{\Psi}^{(0)}
        = (\theta^{(0)}, \gamma^{(0)}, p_1^{(0)}, p_2^{(0)}, a_1^{(0)}, r_1^{(0)}, a_2^{(0)}, r_2^{(0)}).
    \]
    \item[Step E.] Given the current parameter estimate $\boldsymbol{\Psi}^{(k)}$, compute the conditional expectations of the latent variables for each subject $i$ based on the observed data $(T_{i1}, T_{i2}, \delta_{i1}, \delta_{i2})$, and evaluate the Q-function
    \[
        Q(\boldsymbol{\Psi}\mid \boldsymbol{\Psi}^{(k)})
        = \E\left[\ell_c(\boldsymbol{\Psi})
        \,\middle|\, \{(T_{i1},T_{i2},\delta_{i1},\delta_{i2})\}_{i=1}^n\right].
    \]

    \item[Step M.] Maximize $Q(\boldsymbol{\Psi}\mid \boldsymbol{\Psi}^{(k)})$ with respect to $\boldsymbol{\Psi}$ to obtain the updated estimate $\boldsymbol{\Psi}^{(k+1)}$. If $\|\boldsymbol{\Psi}^{(k+1)} - \boldsymbol{\Psi}^{(k)}\| \leq \varepsilon$, where $\varepsilon$ is a prespecified tolerance, stop; otherwise, set $k \leftarrow k+1$ and return to Step E.
\end{itemize}

\subsection{Numerical analysis}
We conduct a numerical study to compare the EM algorithm described above with the direct optimization approach in Section 4.1 of the main text. Throughout, we restrict attention to the case $\OR=1$, that is, independent cure indicators. Data are generated under the Gumbel copula model with Weibull marginals (see Section 2.1 of the main text) with sample size $n=200$ and $1,000$ Monte Carlo replications. Two settings are considered: Setting~A with $(\theta,\gamma,p_1,p_2,a_1,r_1,a_2,r_2)=(2,0.5,0.6,0.4,1,1.5,1,2)$ and Setting~B with $\theta=0.5$ and all other parameters unchanged. The two methods are compared in terms of estimation accuracy, attained maximized log-likelihood and computational time. Estimation accuracy is assessed using the bias, standard deviation (SD), and mean squared error (MSE) of the parameter estimates.

\begin{table}[ht]
    \centering
    \caption{Bias, SD, and MSE of the parameter estimates under the EM algorithm and direct optimization, based on $1,000$ Monte Carlo replications with sample size $n=200$; the upper and lower panels correspond to Setting~A ($\theta=2$) and Setting~B ($\theta=0.5$), respectively}
    \label{tab:bias-sd-mse}
    \renewcommand{\arraystretch}{1.15}
    \setlength{\tabcolsep}{5pt}
    \begin{tabular}{c c cccc cccc}
        \toprule
        & & \multicolumn{4}{c}{EM algorithm} & \multicolumn{4}{c}{Direct optimization} \\
        \cmidrule(lr){3-6} \cmidrule(lr){7-10}
        Parameter & True & Mean & Bias & SD & MSE & Mean & Bias & SD & MSE \\
        \midrule
        $\theta$ & 2.0 & 2.095 & 0.095 & 0.622 & 0.396 & 2.098 & 0.098 & 0.623 & 0.398 \\
        $\gamma$ & 0.5 & 0.622 & 0.122 & 0.521 & 0.286 & 0.617 & 0.117 & 0.521 & 0.285 \\
        $p_1$    & 0.6 & 0.596 & $-0.004$ & 0.045 & 0.002 & 0.596 & $-0.004$ & 0.045 & 0.002 \\
        $p_2$    & 0.4 & 0.395 & $-0.005$ & 0.048 & 0.002 & 0.395 & $-0.005$ & 0.048 & 0.002 \\
        $a_1$    & 1.0 & 1.038 & 0.038 & 0.159 & 0.027 & 1.037 & 0.037 & 0.159 & 0.027 \\
        $r_1$    & 1.5 & 1.678 & 0.178 & 0.590 & 0.379 & 1.675 & 0.175 & 0.590 & 0.378 \\
        $a_2$    & 1.0 & 1.035 & 0.035 & 0.154 & 0.025 & 1.034 & 0.034 & 0.154 & 0.025 \\
        $r_2$    & 2.0 & 2.283 & 0.283 & 0.932 & 0.948 & 2.278 & 0.278 & 0.931 & 0.942 \\
        \midrule
        $\theta$ & 0.5 & 0.539 & 0.039 & 0.300 & 0.091 & 0.541 & 0.041 & 0.299 & 0.091 \\
        $\gamma$ & 0.5 & 0.606 & 0.106 & 0.522 & 0.284 & 0.601 & 0.101 & 0.523 & 0.284 \\
        $p_1$    & 0.6 & 0.595 & $-0.005$ & 0.048 & 0.002 & 0.595 & $-0.005$ & 0.048 & 0.002 \\
        $p_2$    & 0.4 & 0.395 & $-0.005$ & 0.050 & 0.003 & 0.395 & $-0.005$ & 0.050 & 0.003 \\
        $a_1$    & 1.0 & 1.040 & 0.040 & 0.172 & 0.031 & 1.038 & 0.038 & 0.171 & 0.031 \\
        $r_1$    & 1.5 & 1.674 & 0.174 & 0.605 & 0.396 & 1.669 & 0.169 & 0.598 & 0.386 \\
        $a_2$    & 1.0 & 1.031 & 0.031 & 0.149 & 0.023 & 1.029 & 0.029 & 0.149 & 0.023 \\
        $r_2$    & 2.0 & 2.233 & 0.233 & 0.821 & 0.727 & 2.225 & 0.225 & 0.811 & 0.707 \\
        \bottomrule
    \end{tabular}
\end{table}

\Cref{tab:bias-sd-mse} reports the bias, SD, and MSE of the parameter estimates under the two methods. Both methods converge successfully in all $1,000$ replications under both settings. The resulting estimates are nearly identical: the bias, SD, and MSE agree to the third decimal place for most parameters, confirming that the EM algorithm and direct optimization reach essentially the same maximum of the observed-data log-likelihood. Indeed, the maximum absolute difference in attained log-likelihood across the $1,000$ replications is $0.24$ under Setting~A and $0.90$ under Setting~B, both negligible relative to the typical log-likelihood values of approximately $-300$.

\begin{table}[ht]
    \centering
    \caption{Median CPU time (in seconds) and median number of iterations to convergence under the EM algorithm and direct optimization for sample size $n=200$, based on $1,000$ Monte Carlo replications}
    \label{tab:cputime-iter}
    \renewcommand{\arraystretch}{1.15}
    \setlength{\tabcolsep}{8pt}
    \begin{tabular}{clcc}
        \toprule
        Setting & Method & Median CPU time & Median iterations \\
        \midrule
        A ($\theta=2$)   & EM algorithm        & 1.458 & 26 \\
                         & Direct optimization & 0.028 & --- \\
        \midrule
        B ($\theta=0.5$) & EM algorithm        & 1.290 & 32 \\
                         & Direct optimization & 0.025 & --- \\
        \bottomrule
    \end{tabular}
\end{table}

\Cref{tab:cputime-iter} shows that the EM algorithm is approximately $50$ times slower than direct optimization under both settings. This additional computational burden is attributable to the repeated evaluation of conditional expectations in the E-step together with the numerical maximization of the Q-function in the M-step at each iteration.

Taken together, these results suggest that, for the proposed model, the direct optimization approach is preferable to the EM algorithm. The two methods yield essentially identical parameter estimates and maximized log-likelihood values, but the EM algorithm requires substantially greater computational effort.

\bibliography{ref}